\documentclass[10pt]{article}
\usepackage{amsfonts,amsmath,amssymb,amsthm}
\usepackage[pdftex]{graphicx}
\usepackage[hmargin=1in,vmargin=1in]{geometry}
\usepackage{natbib} 
\usepackage{setspace}
\usepackage{enumerate}
\usepackage[toc,title,titletoc,header]{appendix}
\usepackage[colorlinks,citecolor=blue,hypertexnames=false]{hyperref}
\usepackage{etoolbox}
\usepackage{booktabs}
\usepackage{bbm}
\usepackage{tikz,pgfplots}
\allowdisplaybreaks

\def\qed{\rule{2mm}{2mm}}
\def\indep{\perp \!\!\! \perp}

\newtheorem{theorem}{Theorem}[section]
\newtheorem{lemma}{Lemma}[section]

\newtheorem{corollary}{Corollary}[section]

\theoremstyle{definition}
\newtheorem{example}{Example}[section]
\newtheorem{remark}{Remark}[section]
\newtheorem{assumption}{Assumption}[section]

\AtEndEnvironment{remark}{~\qed}
\AtEndEnvironment{example}{~\qed}

\DeclareMathOperator*{\argmax}{argmax}
\DeclareMathOperator*{\argmin}{argmin}

\begin{document}

\title{On the Identifying Power of Generalized Monotonicity \\ for Average Treatment Effects\thanks{We thank the associate editor and three anonymous referees, as well as Sukjin Han, Sokbae Lee, Derek Neal, Kirill Ponomarev, Vitor Possebom, Bernard Salani\'e and Panos Toulis for helpful comments.  Shaikh acknowledges financial support from National Science Foundation Grant SES-2419008.  Vytlacil acknowledges financial support from the Tobin Center for Economic Policy at Yale. Moon acknowledges financial support from the National Science Foundation Graduate Research Fellowship under Grant No. 2141064.}}

\author{Yuehao Bai \\
Department of Economics \\
University of Southern California \\
\url{yuehao.bai@usc.edu}
\and
Shunzhuang Huang \\
Booth School of Business \\
University of Chicago \\
\url{shunzhuang.huang@chicagobooth.edu}
\and
Sarah Moon\\
Department of Economics \\
MIT\\
\url{sarahmn@mit.edu}
\and
Azeem M.\ Shaikh \\
Department of Economics \\
University of Chicago \\
\url{amshaikh@uchicago.edu}
\and
Edward J.\ Vytlacil \\
Department of Economics \\
Yale University \\
\url{edward.vytlacil@yale.edu}
}

\bigskip

\maketitle

\vspace{-0.3in}

\begin{spacing}{1.2}
\begin{abstract}
In the context of a binary outcome, treatment, and instrument, \cite{balke1993nonparametric,balke1997bounds} establish that the monotonicity condition of \cite{imbens1994identification} has no identifying power beyond instrument exogeneity for average potential outcomes and average treatment effects in the sense that adding it to instrument exogeneity does not decrease the identified sets for those parameters whenever those restrictions are consistent with the distribution of the observable data. This paper shows that this phenomenon holds in a broader setting with a multi-valued outcome, treatment, and  instrument, under an extension of the  monotonicity condition  that we refer to as generalized monotonicity. We further show that this phenomenon holds for any restriction on treatment response that is stronger than generalized monotonicity provided that these stronger restrictions do not restrict potential outcomes. Importantly, many models of potential treatments previously considered in the literature imply generalized monotonicity, including the types of monotonicity restrictions considered by \cite{kline2016evaluating}, \cite{kirkeboen2016field}, and \cite{heckman2018unordered}, and the restriction that treatment selection is determined by particular classes of additive random utility models.  We show through a series of examples that restrictions on potential treatments can provide identifying power beyond instrument exogeneity for average potential outcomes and average treatment effects when the restrictions imply that the generalized monotonicity condition is violated.  
In this way, our results shed light on the types of restrictions required for help in identifying average potential outcomes and average treatment effects.

\end{abstract}
\end{spacing}



\thispagestyle{empty} 
\newpage
\setcounter{page}{1}

\section{Introduction} \label{sec:intro}

In their analysis of a setting with a binary outcome, treatment, and instrument, \cite{balke1993nonparametric,balke1997bounds} establish that the monotonicity condition of \cite{imbens1994identification} has no identifying power beyond instrument exogeneity for average potential outcomes and the average treatment effect (ATE).  Here, by no identifying power beyond instrument exogeneity, we mean that adding the monotonicity condition of \cite{imbens1994identification} to instrument exogeneity does not decrease the identified sets for those parameters whenever those restrictions are consistent with the distribution of the observable data.\footnote{For settings with possibly non-binary outcomes, \cite{kitagawa2021identification} shows this phenomenon continues to hold for any parameter that is a function of the marginal distributions of potential outcomes.} In this way, their results contrast with the analysis of  \cite{imbens1994identification}, who showed that their monotonicity condition and instrument exogeneity permitted identification of the local average treatment effect (LATE).  This paper studies the extent to which this phenomenon holds in the broader context of a multi-valued outcome, treatment, and instrument. 

We show that a generalization of the monotonicity condition of \cite{imbens1994identification} to this richer setting also has no identifying power beyond instrument exogeneity for average potential outcomes and ATEs. We hereafter refer to this condition more succinctly as generalized monotonicity.  We show further that this result remains true for any restriction that is in fact \emph{stronger} than generalized monotonicity provided that these stronger restrictions do not restrict potential outcomes in a sense that we will make precise later.  This feature of our results is remarkable because one might expect stronger restrictions, possibly with very complicated restrictions on potential treatments that are difficult to fully characterize, to reduce the identified set at least in some instances, but we show that this is not the case.  Using this result, our analysis accommodates many examples of restrictions on potential treatments that have been previously considered in the literature.  In particular, we show that encouragement designs, the types of monotonicity restrictions considered by \cite{kline2016evaluating}, \cite{kirkeboen2016field}, and \cite{heckman2018unordered}, and certain additive random utility models, including some studied in \cite{lee2023treatment}, all satisfy generalized monotonicity.

In establishing our results, we derive the identified sets for average potential outcomes under any such restriction and instrument exogeneity while maintaining the assumption that these restrictions are consistent with the distribution of the observable data.  Our derivations reveal that the form of the resulting identified sets parallels the form of those derived by \cite{balke1993nonparametric,balke1997bounds} for a binary outcome, treatment, and instrument.  An implication of the form of the identified sets is that average potential outcomes and ATEs are \emph{only} identified under an identification-at-infinity-type condition when imposing instrument exogeneity and any such restriction.  In our analysis, we also derive the identified sets for average potential outcomes and ATEs when imposing instrument exogeneity alone whenever the distribution of the observable data is consistent with instrument exogeneity and generalized monotonicity.  As we explain in Example \ref{ex:rct}, this consistency is necessarily satisfied, for example, in the context of a multi-arm randomized controlled trial with one-sided non-compliance when defining the instrument to be random assignment to a given treatment arm.

Our results further provide necessary conditions on restrictions on potential treatments to help in identifying average potential outcomes and ATEs. See Theorem 3.3 and the subsequent discussion for details. We illustrate this phenomenon through a series of examples of models that need not satisfy generalized monotonicity and have identifying power for average potential outcomes and ATEs.

Our paper differs from the closely related literature that, in the context of a binary outcome, treatment, and instrument, considers the identifying power of the monotonicity condition of \cite{imbens1994identification} and instrument exogeneity for the distribution (as opposed to the average) of potential outcomes, or considers the identifying power of these conditions when combined with additional restrictions on potential outcomes.   In particular,  \cite{kamat2019identifying}  shows that the monotonicity condition of \cite{imbens1994identification} does have identifying power beyond instrument exogeneity for the (joint) distribution of potential outcomes.   \cite{machado2019instrumental} show that the monotonicity condition of \cite{imbens1994identification} does have additional identifying power for the ATE beyond instrument exogeneity if one additionally imposes an assumption that requires potential outcomes to vary monotonically with the treatment.  Thus, the phenomenon we explore is sensitive to both the choice of parameter and to whether one imposes  assumptions on potential outcomes.

The remainder of the paper is organized as follows. Section \ref{sec:setup} introduces our formal setup, notation and assumptions, including our generalized monotonicity condition.  Our main identification results are presented in  Section \ref{sec:main}.   In Section \ref{sec:ex}, we provide several examples of restrictions on potential treatments that imply our generalized monotonicity condition, and are thus examples of restrictions that have no identifying power beyond instrument exogeneity for average potential outcomes or ATEs. In contrast, in Section \ref{sec:ex_not}, we provide several examples of restrictions on potential treatments that imply that our generalized monotonicity condition does not hold, and further show that some of these restrictions in fact have identifying power beyond instrument exogeneity for average potential outcomes and ATEs. Proofs of all results can be found in the Appendix.

\section{Setup and Notation}
\label{sec:setup}

Denote by $Y \in \mathcal Y$ a multi-valued outcome of interest, by $D \in \mathcal D$ a multi-valued endogenous regressor, and by $Z \in \mathcal Z$ a multi-valued instrumental variable.\footnote{Our restriction to a multi-valued $Y$ facilitates exposition, but is not essential.  At the expense of slightly more complicated arguments, we can accommodate more generally any real-valued $Y$.}  To rule out degenerate cases, we assume throughout that $2 \leq |\mathcal Y| < \infty$, $2 \leq |\mathcal D| < \infty$, and $2 \leq |\mathcal Z| < \infty$.  Further denote by $Y_d \in \mathcal Y$ the potential outcome if $D = d \in \mathcal D$ and by $D_z \in \mathcal D$ the potential treatment if $Z = z \in \mathcal Z$.  We impose the usual consistency assumption, 
\begin{equation}\label{eq:potential}
Y = \sum_{d \in \mathcal D} Y_d \mathbbm 1\{D = d\} \text{~~and~~} D = \sum_{z \in \mathcal Z} D_z \mathbbm 1\{Z = z\}~. 
\end{equation}
Let $P$ denote the distribution of $(Y, D, Z)$ and $Q$ denote the distribution of $((Y_d : d \in \mathcal D),(D_z : z \in \mathcal Z), Z)$. Note that \eqref{eq:potential} defines a mapping $T$ through
\begin{equation} \nonumber
(Y, D, Z) = T((Y_d : d\in \mathcal D),(D_z : z \in \mathcal Z), Z)~,
\end{equation}
and therefore $P = QT^{-1}$. In what follows, we will say that a given $Q$ rationalizes a given $P$ if $P = QT^{-1}$.

Below we will require that $Q \in \mathbf Q$, where $\mathbf Q$ is a class of distributions satisfying assumptions that we will specify.   Different choices of $\mathbf Q$ represent different assumptions that we impose on the distribution of potential outcomes and potential treatments.  In this sense, $\mathbf Q$ may be viewed as a model for potential outcomes and potential treatments.

Given $P$ and a model $\mathbf Q$, the set of $Q \in \mathbf Q$ that can rationalize $P$ is
\[ \mathbf Q_0(P, \mathbf Q) = \{Q \in \mathbf Q: P = Q T^{-1}\}~, \] i.e., the pre-image of $P$ under $T$.
We say $\mathbf Q$ is consistent with $P$ if and only if $\mathbf Q_0(P, \mathbf Q ) \ne \emptyset$. We will start by considering models $\mathbf Q$ for which every $Q \in \mathbf Q$ satisfies
\begin{assumption}[Instrument Exogeneity] \label{as:exog}
$((Y_d :  d \in \mathcal D),(D_z : z \in \mathcal Z)) \indep Z$ under $Q$.
\end{assumption}
\noindent Our final result on the identifying power of generalized monotonicity will also apply to the weaker exogeneity restriction in \cite{richardson2013single} that avoids the ``cross-world'' restrictions of Assumption \ref{as:exog}; see Assumption \ref{as:marginal-exog} and Corollary \ref{cor:marginal} below in Section \ref{sec:ifconsistent}. If $Q$ satisfies Assumption \ref{as:exog}, then
\begin{equation} \label{eq:p-q}
p_{yd|z} := P \{Y = y, D = d \mid Z = z\}= Q \{Y_d = y, D_z = d \mid Z = z\} = Q \{Y_d = y, D_z = d\}~.
\end{equation}
Since the marginal distribution of $Z$ under $P$ and $Q$ are the same, i.e., for all $z \in \mathcal Z$,
\[ P\{Z = z\} = Q\{Z = z\}~, \] $P = QT^{-1}$ if and only if \eqref{eq:p-q} holds. Thus, if all $Q \in \mathbf Q$ satisfies Assumption \ref{as:exog}, then $\mathbf Q_0(P, \mathbf Q)$ can be simplified as
\begin{equation} \label{eq:Q0PQ}
\mathbf Q_0(P, \mathbf Q ) = \bigg\{Q \in \mathbf Q: ~p_{yd|z} = Q \{Y_d = y, D_z = d\} \text{ for all }  y \in \mathcal Y, d \in \mathcal D, z \in \mathcal Z \bigg\}~.
\end{equation}

Let $\theta(Q) = ( \mathbb{E}_Q[Y_d] : d \in \mathcal{D} )$ denote the vector of average potential outcomes.  For fixed $P$ and $\mathbf Q$, the identified set for $\theta(Q)$ under $P$ relative to $\mathbf{Q}$ is given by 
\begin{equation*}
\Theta_0(P, \mathbf Q) := \{ \theta(Q) : Q \in \mathbf Q_0(P, \mathbf Q)\} ~. 
\end{equation*}
$\Theta_0(P, \mathbf Q) $ is nonempty  whenever $\mathbf Q_0(P, \mathbf Q)$ is nonempty.  By construction, this set is ``sharp'' in the sense that for any value in the set there exists $Q \in \mathbf Q_0(P, \mathbf Q)$ for which $\theta(Q)$ equals the prescribed value. The identified set for $\theta(Q)$ immediately implies that the identified set for any parameter $\lambda = \lambda(\theta)$ is given by $\lambda(\Theta_0(P, \mathbf Q))$. An important example is $\mathbb{E}_Q[Y_j]- \mathbb{E}_Q[Y_k]$, the ATE for treatment $j$ versus treatment $k$.  
 
In the next section, we consider identification in any model of potential treatments that implies the following restriction on all $Q \in \mathbf{Q}$:

\begin{assumption}[Generalized Monotonicity] \label{as:encouragement}
For each $d \in \mathcal{D}$, there exists $z^\ast = z^{\ast}(d, Q) \in \mathcal{Z}$ such that 
\begin{equation}\label{eq:encouragement} Q\{D_{z^\ast} \ne d ,   ~D_{z^{\prime}}=d ~ \mbox{for some}  ~ z^{\prime} \ne z^\ast \}=0~.\end{equation}
\end{assumption}

In what follows, we refer to Assumption \ref{as:encouragement} as generalized monotonicity. It states that under $Q$, for each treatment status $d \in \mathcal D$, there exists a value (possibly depending on $d$ and $Q$) of the instrument $z^\ast \in \mathcal Z$ that maximally encourages all individuals to $d$. Here, by ``maximally encourage'', we mean that if an individual does not choose $d$ when $Z = z^\ast$, then they never choose $d$ for any other value of $Z$.  Equivalently, if an individual chooses $d$ when $Z$ is equal to any value other than $z^\ast$, then they have to choose $d$ when $Z = z^\ast$.  When $\mathcal{D}= \mathcal Z = \{0,1\}$, Assumption \ref{as:encouragement} is equivalent to the monotonicity assumption of \cite{imbens1994identification}.

We emphasize that Assumption \ref{as:encouragement} only requires, for each possible value of the treatment, that there exists a value of the instrument that maximally encourages that treatment; it does not require that the value of the instrument is unique. For a given distribution $Q$ and given treatment $d \in \mathcal D$, let $\mathcal{Z}^\ast(d, Q)$ denote the set of $z^\ast$ that satisfy \eqref{eq:encouragement}.  In this notation, Assumption \ref{as:encouragement} can be restated as $\mathcal{Z}^\ast(d, Q) \neq \emptyset$ for each $d \in \mathcal{D}$. In the statement of Assumption \ref{as:encouragement}, $z^\ast(d, Q)$ is allowed to change across $Q$. The following lemma shows that $\mathcal{Z}^\ast(d, Q)$ is identified from $P$ and is hence the same for all $Q$ that rationalizes $P$ and satisfies Assumptions \ref{as:exog} and \ref{as:encouragement}. In what follows, we will therefore write $\mathcal Z^\ast(d)$ and $z^\ast(d)$ whenever the given distribution $Q$ rationalizes $P$ and satisfies Assumptions \ref{as:exog} and \ref{as:encouragement}. This result generalizes the corresponding result in  \cite{imbens1994identification}.

\begin{lemma} \label{lem:maxp}
    Suppose $Q$ satisfies Assumptions \ref{as:exog} and \ref{as:encouragement} and $P = Q T^{-1}$. Then, $z \in \mathcal Z^\ast(d, Q)$ if and only if
    \begin{equation} \label{eq:zd-id}
    P \{ D = d \mid Z=z \} \ge P \{ D = d \mid Z=z'  \} \text{ for all } z' \in \mathcal Z~.
    \end{equation}
\end{lemma}

Below we prove the necessity of \eqref{eq:zd-id}; sufficiency is established in the appendix. Note that, for any $d \in \mathcal{D}$, $z' \in \mathcal{Z}$, and any  $z \in \mathcal Z^\ast(d, Q)$,
\begin{align*} P \{ D = d \mid Z=z \}  =& ~  Q \{ D_z=d\}\\
=&  ~ Q \{D_z=d, D_{z'}=d\} + Q \{D_z=d, D_{z'} \ne d\}\\
=&  ~ Q \{D_{z'}=d\} + Q \{D_z=d, D_{z'} \ne d\}\\
\ge & ~ Q \{D_{z'}=d\}\\
= & ~ P \{ D = d \mid Z=z'  \}~,
\end{align*}
where the first and last equalities exploit Assumption  \ref{as:exog}, and the third equality uses Assumption  \ref{as:encouragement}.

\section{Main Result} 
\label{sec:main}

In order to describe our main result, we first introduce some further notation.  Denote by $\mathbf Q_E^\ast$ (where $E$ stands for exogeneity) the set of all distributions that satisfy Assumption \ref{as:exog} and by $\mathbf Q_{E,M}^\ast$ (where $M$ stands for generalized monotonicity) the set of all distributions that satisfy Assumptions \ref{as:exog} and \ref{as:encouragement}. We will further require that the model does not restrict potential outcomes in the following sense:
\begin{assumption}[Unrestricted Potential Outcomes] \label{as:norestrictY}
Let $Q \in \mathbf Q$ and $Q' \in \mathbf Q_E^\ast$. If the distributions of $(D_z: z \in \mathcal Z)$ under $Q$ and $Q'$ are the same, then $Q' \in \mathbf Q$.
\end{assumption}

In terms of this notation, our main result can be stated as follows:

\begin{theorem} \label{theorem:idpower}
Suppose $\mathbf Q \subseteq \mathbf Q_{E,M}^\ast$ and $\mathbf Q$ satisfies Assumption \ref{as:norestrictY}.
Then, for any $P$ such that $\mathbf Q_0(P, \mathbf Q) \neq \emptyset$, we have $\Theta_0(P, \mathbf Q) = \Theta_0(P, \mathbf Q_{E,M}^\ast) = \Theta_0(P, \mathbf Q_E^\ast)$.
\end{theorem}

\noindent Theorem \ref{theorem:idpower} describes the sense in which restrictions on potential treatments stronger than generalized monotonicity have no identifying power for average potential outcomes and ATEs provided that these stronger restrictions do not restrict potential outcomes.  This result is established through Theorems \ref{theorem:ev} and \ref{theorem:same} below.  Theorem \ref{theorem:ev}, developed in Section \ref{sec:idsets}, characterizes $\Theta_0(P, \mathbf Q)$ for any model $\mathbf Q$ that is stronger than Assumptions \ref{as:exog} and \ref{as:encouragement}, i.e., $\mathbf Q \subseteq \mathbf Q_{E,M}^\ast$, and does not restrict potential outcomes in the sense of Assumption \ref{as:norestrictY}.  The result shows, in particular, that $\Theta_0(P, \mathbf Q) = \Theta_0(P, \mathbf Q_{E,M}^\ast)$ for any such model $\mathbf Q$ whenever $\mathbf Q_0(P, \mathbf Q) \neq \emptyset$.  Remarkably, this result holds even if the model $\mathbf Q$ is \emph{strictly} more restrictive than Assumptions \ref{as:exog} and \ref{as:encouragement} in the sense that  $\mathbf Q \subsetneqq \mathbf Q_{E,M}^\ast$.  On the other hand, Theorem \ref{theorem:same} and Corollary \ref{cor:exog}, developed in Section \ref{sec:ifconsistent}, show that $\Theta_0(P, \mathbf Q_{E,M}^\ast) = \Theta_0(P, \mathbf Q_E^\ast)$ whenever $\mathbf Q_0(P, \mathbf Q_{E,M}^\ast) \neq \emptyset$.  Together, these results immediately imply Theorem \ref{theorem:idpower}. In fact, Theorem \ref{theorem:same} shows the stronger result that if a submodel of instrument exogeneity and generalized monotonicity is consistent with $P$, then any model sandwiched between this submodel and the model that only assumes mean independence leads to the same identified set for average potential outcomes and ATEs.  This observation  allows us to establish that generalized monotonicity also has no identifying power for average potential outcomes and ATEs beyond the weaker exogeneity restriction of \cite{richardson2013single}.

\subsection{Identified Sets for $\mathbf Q \subseteq \mathbf Q_{E,M}^\ast$} \label{sec:idsets}
For $d \in \mathcal D$ and $z \in \mathcal Z$, define $\beta_{d|z} = \mathbb{E}_P[ Y \mathbbm{1}\{D=d\} \mid Z=z]$. In addition, define $y^L = \min(\mathcal{Y})$ and $y^U = \max(\mathcal{Y})$.
The following theorem derives the identified set for $\theta(Q)$, relative to any model that assumes instrument exogeneity and generalized monotonicity but does not restrict potential outcomes. Note in particular that the assumptions allow for $\mathbf Q \subsetneqq \mathbf Q_{E,M}^\ast$, in which case the model assumes \emph{strictly} more than instrument exogeneity and generalized monotonicity.

\begin{theorem}\label{theorem:ev}
Suppose $\mathbf Q \subseteq \mathbf Q_{E,M}^\ast$ and $\mathbf Q$ satisfies Assumption \ref{as:norestrictY}.
Then, for any $P$ such that $\mathbf Q_0(P, \mathbf Q) \neq \emptyset$, 
\begin{equation} \label{eq:bounds1}
\Theta_0(P, \mathbf Q) = \prod_{d \in \mathcal{D}} \left[ \beta_{d|z^\ast(d)} + y^L (1- \sum_{y \in \mathcal{Y}} p_{yd|z^\ast(d)}) )   ,~ \beta_{d|z^\ast(d)} + y^U (1- \sum_{y \in \mathcal{Y}} p_{yd|z^\ast(d)}) )  \right]~.
\end{equation}
\end{theorem}

We now describe some intuition for Theorem \ref{theorem:ev}.  Note that the distribution of the data, $P$,  only contains information on the distribution of $Y_d$  for those individuals who would take treatment $d$ for some value of the instrument. On the other hand, it contains no information on the distribution of $Y_d$ for those individuals who would not take that treatment for any value of the instrument.  Assumption \ref{as:encouragement} implies that individuals would take treatment $d$ at some value of the instrument if and only if $D_{z^{\ast}(d)}=d$, i.e., when maximally encouraged to do so.  Assumption \ref{as:exog} implies $ \beta_{d|z^\ast(d)} =    \mathbb{E}_Q[ \mathbbm{1}\{D_{z^{\ast}(d)}=d\} Y_d  ] $, and thus captures all the information from $P$ relevant to $\mathbb{E}_Q[Y_d]$, which is the first part of the lower and upper bounds in  \eqref{eq:bounds1}. In contrast, Assumptions \ref{as:exog} and \ref{as:encouragement} imply that the probability that an individual would not take treatment $d$ for any value of $Z$ is identified from $P$ to be $Q\{D_{z^{\ast}(d)}\ne d\} = 1- \sum_{y \in \mathcal{Y}} p_{yd|z^\ast(d)},$ but $P$ contains no information on the distribution of $Y_d$ for such individuals.  Furthermore, that $\mathbf Q$ satisfies Assumption \ref{as:norestrictY}  implies that the model does not restrict the distribution of $Y_d$ for such individuals beyond $y \in \mathcal{Y}$, so that we can set $Y_d$ to be any value between $y^L$ and $y^U$ for these individuals, which constitutes the second part of the upper and lower bounds in \eqref{eq:bounds1}.

\begin{remark}
Under the instrument exogeneity and  monotonicity assumptions of  \cite{imbens1994identification},
\cite{balke1993nonparametric,balke1997bounds} found the same form of the identified set for $\theta(Q)$ as \eqref{eq:bounds1} when $\mathcal{Y}=\mathcal{D}=\mathcal{Z}=\{0,1\}$. Theorem \ref{theorem:ev} therefore generalizes the result of \cite{balke1993nonparametric,balke1997bounds} to more than two treatment arms and instrument values,  to outcomes taking more than two values, and, more surprisingly, to show that the same identified set holds when imposing possibly \emph{stronger} restrictions on potential treatments than generalized monotonicity.
\end{remark}

Theorem \ref{theorem:ev} immediately implies the following result on the identified sets for the ATE of treatment $j$ versus $k$:

\begin{corollary}\label{cor:ev2}
Under the assumptions of Theorem \ref{theorem:ev},  the identified set for $\mathbb E_Q[Y_j -Y_k]$ is  given by:
  \begin{multline} \label{eq:bounds2} 
  \biggl[
(\beta_{j|z^\ast(j)}- \beta_{k|z^\ast(k)}) + (y^L -y^U) 
  + y^U \sum_{y \in \mathcal{Y}}  p_{yk|z^\ast(k)} - y^L \sum_{y \in \mathcal{Y}} p_{yj|z^\ast(j)}~, \\
 ~ 
    (\beta_{j|z^\ast(j)}- \beta_{k|z^\ast(k)}) + (y^U-y^L) 
  + y^L \sum_{y \in \mathcal{Y}} p_{yk|z^\ast(k)} - y^U \sum_{y \in \mathcal{Y}} p_{yj|z^\ast(j)}
  \biggr]~. \end{multline}
\end{corollary}
e

\begin{remark}
From Corollary \ref{cor:ev2}, the width of the identified set for $\mathbb E_Q[Y_j -Y_k]$ under the assumptions of Theorem \ref{theorem:ev} is given by $$(y^U-y^L) \left(P\{D \ne j \mid Z= z^\ast(j)\} + P\{D \ne k \mid Z= z^\ast(k)\} \right)~,$$ which, following Lemma \ref{lem:maxp}, equals $$(y^U-y^L) \left( \min_{z \in \mathcal{Z}} P\{D \ne j \mid Z= z\} +\min_{z \in \mathcal{Z}} P\{D \ne k \mid Z= z\} \right)~. $$ Therefore, when imposing generalized monotonicity, as well as when imposing any restriction implying generalized monotonicity, the ATE of $j$ versus $k$ is \emph{only} identified ``at infinity'' \citep{heckman1990varieties,andrews1998semiparametric} in the sense that identification requires
\begin{equation}\label{eq:ident_infinity}
    \min_{z \in \mathcal{Z}} P\{D \ne j \mid Z= z\}= \min_{z \in \mathcal{Z}} P\{D \ne k \mid Z= z\}=0~.
\end{equation} In other words, identification of the the ATE of $j$ versus $k$ under generalized monotonicity or under any restriction implying generalized monotonicity requires that there is some value of the instrument such that everyone takes treatment $j$ at that value of the instrument, and some value of the instrument such that everyone takes treatment $k$ at that value of the instrument.  In contrast, by imposing restrictions on potential treatments that imply that generalized monotonicity is violated, the ATE can sometimes be identified without \eqref{eq:ident_infinity} even when potential outcomes are unrestricted; see Example \ref{eg:alt} in Section \ref{sec:ex_not} below.
\end{remark}
 
\subsection{Identifying Power of Generalized Monotonicity} \label{sec:ifconsistent}
Theorem \ref{theorem:ev} above establishes that, for possibly multi-valued $Y$, $D$, and $Z$, and any model $\mathbf Q \subseteq \mathbf Q_{E,M}^\ast$ such that $\mathbf Q$ does not restrict the potential outcomes, if $\mathbf Q_0(P, \mathbf Q) \neq \emptyset$, then $\Theta_0(P, \mathbf Q)$ equals \eqref{eq:bounds1}. We now show that $\Theta_0(P, \mathbf Q_{E,M}^\ast) = \Theta_0(P, \mathbf Q_E^\ast)$ as long as $\mathbf Q_0(P, \mathbf Q_{E,M}^\ast) \neq \emptyset$, so that the identified set for $\theta(Q)$ assuming instrument exogeneity and generalized monotonicity coincides with the identified set assuming instrument exogeneity alone, as long as both assumptions are consistent with the distribution of the observed data. Our result therefore generalizes \cite{balke1993nonparametric,balke1997bounds}, which study the case of binary $Y$, $Z$, and $D$.

In order to do so, we consider a mean independence assumption even weaker than Assumption \ref{as:exog}, and show the identified set under this even weaker assumption is also \eqref{eq:bounds1}. A sandwich argument will then lead to our desired result. In particular, we first establish that the identified set for $\theta(Q)$ in \eqref{eq:bounds1} coincides with the identified set under the weaker mean independence assumption considered by 
\cite{robins1989analysis} and \cite{manski1990nonparametric}:
\begin{assumption}[Mean Independence] \label{as:mean_iv}
$\mathbb{E}_Q[Y_d \mid Z=z ]= \mathbb{E}_Q[Y_d]$ for all $d \in \mathcal{D}$ and $z \in \mathcal{Z}$.
\end{assumption}
Note Assumption \ref{as:mean_iv} is weaker than instrument exogeneity in Assumption \ref{as:exog}, and does not imply  \eqref{eq:p-q}. Let $\mathbf Q_{\it MI}^\ast$ denote the set of all $Q$ that satisfies Assumption \ref{as:mean_iv} (where $MI$ stands for mean independence). Following \cite{robins1989analysis} and \cite{manski1990nonparametric}, the following lemma derives the identified set for $\theta(Q)$ under mean independence:

\begin{lemma} \label{lem:mean_iv}
    Suppose $\mathbf Q_0(P, \mathbf Q_{\it MI}^\ast) \neq \emptyset$. Then,
    \begin{equation}\label{eq:meaniv2}
    \Theta_0(P, \mathbf Q_{\it MI}^\ast) = \prod_{d \in \mathcal{D} } \left[  \max_{z \in \mathcal{Z}}  \{ \beta_{d|z}
  + y^L (1- \sum_{y \in \mathcal{Y}} p_{yd \mid z} )   \} , ~ \min_{z \in \mathcal{Z}} \{ \beta_{d|z}
  + y^U (1- \sum_{y \in \mathcal{Y}} p_{yd \mid z} )  \}  \right].
\end{equation}
\end{lemma}
  
The following lemma, which relies on the observation in Lemma \ref{lem:maxp}, establishes the equivalence between the identified sets in \eqref{eq:bounds1} and \eqref{eq:meaniv2} when $\mathbf Q_0(P, \mathbf Q_{E,M}^\ast) \neq \emptyset$.
 
\begin{lemma}\label{lemma:ev2}
   Suppose $\mathbf Q \subseteq \mathbf Q_{E, M}^\ast$, $\mathbf Q$ satisfies Assumption \ref{as:norestrictY}, and $P$ is such that $\mathbf Q_0(P, \mathbf Q) \neq \emptyset$. Then, the sets in \eqref{eq:bounds1} and \eqref{eq:meaniv2}  coincide.
\end{lemma}

Using Lemma \ref{lemma:ev2}, we are able to establish our desired result, which asserts that, maintaining Assumption \ref{as:exog}, additionally imposing Assumption \ref{as:encouragement} either causes the identified set for $\theta(Q)$ to become empty (if those assumptions are not consistent with $P$) or leaves the identified set for $\theta(Q)$ unchanged (if those assumptions are consistent with $P$). In fact, we will establish a stronger result, that if a submodel of instrument exogeneity and generalized monotonicity is consistent with $P$, then any model sandwiched between this submodel and the model that only assumes mean independence leads to the same identified set for $\theta(Q)$. 

\begin{theorem}\label{theorem:same}
Suppose $\mathbf Q \subseteq \mathbf Q_{E, M}^\ast$ and $\mathbf Q$ satisfies Assumption \ref{as:norestrictY}. Further suppose $\mathbf Q'$ satisfies $$\mathbf Q \subseteq \mathbf Q' \subseteq \mathbf Q_{\it MI}^\ast~.$$ Then, for any $P$ such that $\mathbf Q_0(P, \mathbf Q) \neq \emptyset$, $$\Theta_0(P, \mathbf Q) = \Theta_0(P, \mathbf Q') = \Theta_0(P, \mathbf Q_{\it MI}^\ast)~.$$
\end{theorem}

\begin{remark} \label{remark:contradict}
Theorem \ref{theorem:same} implies that in order for a model to have identifying power for average potential outcomes, it has to be the case that the model does not contain a submodel of instrument exogeneity and generalized monotonicity that is consistent with $P$. In other words, the model has to contradict Assumption \ref{as:exog} or \ref{as:encouragement}. We illustrate this observation in Example \ref{eg:alt} below.
\end{remark}

\begin{corollary} \label{cor:exog}
For any $P$ such that $\mathbf Q_0(P, \mathbf Q_{E,M}^\ast) \neq \emptyset$, $\Theta_0(P, \mathbf Q_{E,M}^\ast) = \Theta_0(P, \mathbf Q_E^\ast)$.
\end{corollary}

\begin{remark} \label{remark:implication}
An implication of Corollary \ref{cor:exog} is that the identified set for $\theta(Q)$ under Assumption \ref{as:exog} alone will be \eqref{eq:bounds1}, regardless of whether Assumption \ref{as:encouragement} is imposed, as long as the distribution  of the data is consistent with Assumptions \ref{as:exog} and \ref{as:encouragement}.   
\end{remark}

Theorem \ref{theorem:same} further implies that under the following weaker exogeneity assumption, generalized monotonicity also has no identifying power for average potential outcomes and ATEs:

\begin{assumption}[Weak Instrument Exogeneity] \label{as:marginal-exog}
Under $Q$, $(Y_d, D_z ) \indep Z$ for all $  d \in \mathcal D, z \in \mathcal Z$.
\end{assumption}
\noindent See \cite{richardson2013single} for an analysis of alternative exogeneity restrictions, and in particular how Assumption \ref{as:marginal-exog} avoids the ``cross-world'' restrictions of the stronger joint independence in Assumption \ref{as:exog}. Denote by $\mathbf Q_{\it WE}^\ast$ the set of all distributions that satisfy Assumption \ref{as:marginal-exog} and $\mathbf Q_{\it WE, M}^\ast$ the set of all distributions that satisfy Assumptions \ref{as:marginal-exog} and \ref{as:encouragement}.

\begin{corollary} \label{cor:marginal}
For any $P$ such that $\mathbf Q_0(P, \mathbf Q_{\it WE,M}^\ast) \neq \emptyset$, $\Theta_0(P, \mathbf Q_{\it WE,M}^\ast) = \Theta_0(P, \mathbf Q_{\it WE}^\ast)$.
\end{corollary}


\section{Examples of Models That Satisfy Assumption \ref{as:encouragement}} 
\label{sec:ex}

We now consider some restrictions on potential treatments that have been considered previously in the literature.  In each case, we show $\mathbf Q \subseteq \mathbf Q_{E,M}^\ast$; in particular, these restrictions satisfy generalized monotonicity.  We emphasize that frequently $\mathbf Q \subsetneqq \mathbf Q_{E,M}^\ast$. In what follows, it is implicitly understood that Assumption \ref{as:norestrictY} is satisfied, so that the model imposes no restriction on potential outcomes.  Thus, in each example Theorem \ref{theorem:idpower} applies and such restrictions do not provide any identifying power for average potential outcomes or ATEs.  In Appendix C, we consider three additional examples: an RCT with a ``close substitute'' as considered in \cite{kline2016evaluating}; the monotonicity and ``irrelevance'' assumptions considered in \cite{kirkeboen2016field}; and an additive random utility model for a binary treatment.

\begin{example} \label{ex:rct}
Consider a multi-arm randomized controlled trial (RCT) with noncompliance, where $Z=d$ denotes random assignment to treatment  $d$,  $D_d=d$ denotes that the subject would comply with  assignment if assigned to  treatment  $d$, and $|\mathcal D|= |\mathcal Z|$. More generally, not necessarily in the context of an RCT,  one can interpret $Z=d$ as encouragement to treatment $d$ and interpret $D_d=d$  as the   subject would take treatment $d$ if encouraged to do so. In this example, $Q$ satisfies Assumption \ref{as:exog} because $Z$ is randomly assigned. We may generalize the ``no-defier'' restriction of \cite{angrist1996identification} as: for each $d \in \mathcal D$, \[ Q\{D_{d} \ne d ,   ~D_{d^{\prime}}=d ~ \mbox{for some}  ~ d^{\prime} \ne d \}=0~,\] i.e., there is zero probability that a subject would not take treatment $d$ if assigned to (encouraged to take) $d$ but would take $d$ if assigned (encouraged) to some other treatment $d^{\prime} \ne d$. If $Q$ satisfies this generalized no-defier restriction, then Assumption \ref{as:encouragement}  holds with $z^\ast(d)=d$ for all $d$.  This no-defier restriction in particular holds in the context of an RCT with ``one-sided non-compliance,'' where we assume
\[ Q\{D_z \in \{0,z\}\}=1~, \] for all $z \in \mathcal Z$. Here, non-compliance is one-sided because one can fall back to the control group if assigned to $d$ but cannot choose $d$ if assigned to the control group.  
 \end{example}

\begin{remark} \label{remark:implication2}
As explained in Example \ref{ex:rct}, in a multi-arm RCT with one-sided noncompliance, $P$ will necessarily be consistent with Assumptions \ref{as:exog} and \ref{as:encouragement}, i.e., $\mathbf Q(P,\mathbf Q_{E,M}) \neq \emptyset$.  Following Remark \ref{remark:implication}, Theorem \ref{theorem:same} therefore implies that the identified set for $\theta(Q)$ under  Assumption \ref{as:exog}  alone will be  \eqref{eq:bounds1} for such a multi-arm RCT.  A further implication is that the identified set under Assumption \ref{as:exog} on $\mathbb{E}[Y_j-Y_k]$ for a multi-arm RCT with one-sided noncompliance depends only on the treatment arms for random assignment to treatments $j$ and $k$.
\end{remark}

\begin{example}\label{ex:chengsmall}
\cite{cheng2006bounds} consider an RCT with noncompliance where $\mathcal{D}=\mathcal{Z}=\{0,1,2\}$. Assumption \ref{as:exog} continues to hold because $Z$ is randomly assigned. In such a setting, they develop bounds on average effects within subgroups defined by potential treatments which, following  the terminology of \cite{frangakis2002principal}, they call ``principal strata.''  While  \cite{bai2025inference} derives the identified sets for those parameters given their assumptions, we now use our analysis to consider instead identification of average potential outcomes and ATEs given their assumptions. Their  ``Monotonicity I'' assumption  is equivalent to one-sided noncompliance in the preceding example. Their ``Monotonicity II'' assumption states that subjects who would comply with assignment to treatment $2$ would also comply with assignment to treatment $1$, so that
\[ Q\{ D_1=1 \mid D_2=2\}=1~. \]
They argue that such an assumption is plausible in a medical context when treatment $1$ has fewer side effects than treatment $2$, and in their application to treatments for alcohol dependence, in which complying with treatment $1$ (compliance enhancement therapy) requires less effort by subjects than complying with treatment $2$ (cognitive behavioural therapy). Because their Monotonicity I restriction implies our Assumption \ref{as:encouragement}, so does imposing both their Monotonicity I and II restrictions.
\end{example}

\begin{example}\label{ex:unordered}
Suppose $Q$ satisfies Assumption \ref{as:exog}. \cite{heckman2018unordered} define ``unordered monotonicity" as the assumption that,  for any $d \in \mathcal{D}$, and any $z, z^{\prime} \in \mathcal{Z}$,
\begin{equation} \label{eq:unordered}
    Q \{\mathbbm{1}\{D_z=d\} \ge \mathbbm{1}\{D_{z^{\prime}}=d\}\} = 1 \text{ or } Q \{\mathbbm{1}\{D_z=d\} \le \mathbbm{1}\{D_{z^{\prime}}=d\}\} = 1~.
\end{equation}
Assumption \ref{as:encouragement} holds for any $Q$ that satisfies \eqref{eq:unordered}. To see this, note that Assumption \ref{as:encouragement} can be expressed as the requirement that for each $d \in \mathcal D$, there exists $z^\ast(d) \in \mathcal Z$ such that \[ Q \{\mathbbm 1 \{D_{z^\ast(d)} = d\} \geq \mathbbm 1 \{D_z = d\}\} = 1 \text{ for all } z \in \mathcal Z~, \]
which is immediately implied by \eqref{eq:unordered}. Note, however, that $\mathcal Z^\ast(d)$ may not be a singleton unless some inequalities in \eqref{eq:unordered} are strict.
\end{example}

\begin{remark}
Although unordered monotonicity implies Assumption \ref{as:encouragement}, the converse is generally false. For example, suppose $\mathcal{Z}=\{0,1,2,3\}$ and $\mathcal{D}=\{0,1\}$. Suppose $ \mathbbm{1}\{D_3=1\} \ge \mathbbm{1}\{D_{z}=1\} $ w.p.1 under $Q$ for $z \ne 3$ and $ \mathbbm{1}\{D_0=0\} \ge \mathbbm{1}\{D_{z}=0\} $ w.p.1 under $Q$ for $z \ne 0$, but $Q\{ D_1=1, D_2=0\}$ and $Q\{ D_1=0, D_2=1\}$ are both strictly positive.  Then Assumption \ref{as:encouragement} holds with $z^\ast(0) = 0$ and  $z^\ast(1) = 3$, but unordered monotonicity fails. In particular, $\mathbbm{1} \{D_1 = 1\}$ and $\mathbbm{1} \{D_2 = 1\}$ are not ordered, thus violating \eqref{eq:unordered}. We thus conclude that if $\mathbf Q$ is defined as the set of distributions that satisfy instrument exogeneity and unordered monotonicity, then $\mathbf Q \subsetneqq \mathbf Q_{E,M}^\ast$.
\end{remark}

\begin{example} \label{ex:arum}
Suppose under $Q$, $(D_z: z \in \mathcal Z)$ is determined by
\begin{equation} \label{eq:ARUM} D_z = \argmax_{d \in \mathcal{D}}~ ( g(z, d) + U_{d} )~,
\end{equation}
for $g: \mathcal Z \times \mathcal D \to \Re$ where $\Re$ is the set of real numbers and a random vector $(U_d: d \in \mathcal D)$, whose distribution is absolutely continuous with respect to the Lebesgue measure on $\Re^{|\mathcal D|}$ and $Z \indep ((U_d: d \in \mathcal D), (Y_d : d \in \mathcal D))$.  Hence, $Q$ satisfies Assumption \ref{as:exog} by construction.  Let $\mathbf{Q}$ denote the set of distributions that are consistent with $(D_z: z \in \mathcal Z)$ being determined by \eqref{eq:ARUM} for some $g$ and $(U_d: d \in \mathcal D)$ satisfying these requirements. The model $\mathbf Q$ is called an additive random utility model (ARUM). A sufficient condition for $Q \in \mathbf Q$ to satisfy Assumption \ref{as:encouragement} is that for each $d \in \mathcal D$ there exists $z^\ast(d) \in \mathcal Z$ such that
\begin{equation} \label{eq:uniform}
    g(z^\ast(d), d) - g(z^\ast(d), d') > g(z, d) - g(z, d') \text{ for all } d' \neq d \text{ and } z \neq z^\ast(d)~.
\end{equation}
We refer to the requirement in \eqref{eq:uniform} as uniform targeting of treatment $d$.  The terminology is intended to reflect that there is a value of the instrument that maximizes the gains (in terms of $g$) of choosing treatment $d$ versus any other treatment $d'$ uniformly across these other possible values of the treatment.  In this sense, that value of the instrument targets treatment $d$ uniformly. We now argue by contradiction that \eqref{eq:uniform} implies \eqref{eq:encouragement}; hence, if \eqref{eq:uniform} holds for all $d \in \mathcal D$, then $Q$ satisfies Assumption \ref{as:encouragement}. To this end, suppose that, with positive probability, $D_{z^\ast(d)} = d' \neq d$ but $D_{z'} = d$ for $z' \neq z^\ast(d)$. Then,
\begin{align*}
    g(z^\ast(d), d') + U_{d'} & \geq g(z^\ast(d), d) + U_d~, \\
    g(z', d) + U_d & \geq g(z', d') + U_{d'}~.
\end{align*}
These two inequalities imply
\[ g(z^\ast(d), d) - g(z^\ast(d), d') \leq g(z', d) - g(z', d')~, \]
which violates \eqref{eq:uniform}. A particular example when $|\mathcal Z| \geq |\mathcal D|$ that satisfies \eqref{eq:uniform} with $z^\ast(d)=d$ after a suitable relabelling is
\begin{equation} \label{eq:arumexample}
g(z, d) = \alpha_d +\beta_d \mathbbm 1 \{z = d\}~,
\end{equation}
with $\beta_d > 0$, so that $Z=d$ strictly increases the latent value of treatment $d$ while leaving the values of the remaining options unchanged.  
\end{example}


\begin{remark}\label{remark:HV}
 The ARUM for a binary treatment is equivalent to the Heckman-Vytlacil nonparametric selection model for a binary treatment considered, e.g., in \cite{heckman1999local,heckman2005structural}, and shown by \cite{vytlacil2002independence} to be equivalent to the monotonicity and exogeneity assumptions of \cite{imbens1994identification}.  
\cite{heckman2001instrumental} show that the Heckman-Vytlacil nonparametric selection model for a binary treatment results in an identified set for the ATE of the form in \eqref{eq:bounds2} and that the model has no identifying power beyond instrument exogeneity for the ATE.   Example C.1 in  Appendix C shows that an ARUM for a binary treatment will satisfy Assumption  \ref{as:encouragement} and therefore the results of this paper nest the results of \cite{heckman2001instrumental}.  Example  \ref{ex:arum}, on the other hand, extends their results to nonparametric selection models for a multi-valued treatment.  For a partial identification analysis of a class of parameters that includes the ATE under a nonparametric selection model for a binary treatment and sometimes imposing additional restrictions, see, e.g., \cite{mogstad2018using}, \cite{han2024computational}, and \cite{marx2024sharp}.
\end{remark}

\begin{example}\label{ex:lee2023}
\cite{lee2023treatment} also consider the ARUM defined by \eqref{eq:ARUM} 
without imposing the uniform targeting of \eqref{eq:uniform} for each treatment. Instead, they impose  an assumption that they refer to as ``strict one-to-one targeting'', in which the set of treatments can be partitioned into a set of treatments $\mathcal{D}^\dagger$ that are ``targeted'' and a set of treatments $\mathcal{D} \setminus \mathcal{D}^\dagger $ that are ``not targeted'' such that 
\begin{enumerate}
    \item For $d \in \mathcal{D} \setminus \mathcal{D}^\dagger$, $g(z, d)$ is the same for all $z \in \mathcal Z$;
    \item  For $d \in \mathcal{D}^\dagger$, there exists $z^\dagger(d)$ such that $g(z^\dagger(d), d) > g(z', d)$ for all $z' \neq z^\dagger(d)$ and such  that $g(z', d)$ takes the same value for all $z' \neq z^\dagger(d)$;
    additionally,  $z^\dagger(d) \neq z^\dagger(d')$ for $d, d^{\prime} \in \mathcal{D}^\dagger$, $d\ne d'$.  
 \end{enumerate}
The terminology ``one-to-one'' stems from the second requirement above.  They further impose that there exists a treatment that is known to be non-targeted. This class of ARUMs is equivalent to imposing  \eqref{eq:arumexample} for targeted treatments, imposing $g(z, d)=\alpha_d$ for non-targeted treatments, and imposing that there is at least one non-targeted treatment. In such a setting, \cite{lee2023treatment} analyze the identification of a particular class of average effects within subgroups defined by potential treatments.  We now use our analysis to consider instead the identification of average potential outcomes and ATEs given their assumptions.  Suppose strict one-to-one targeting holds, and additionally suppose that there is at least one targeted treatment. In Appendix B.1, we argue that Assumption \ref{as:encouragement} holds when  $|\mathcal{Z}| > |\mathcal{D}^\dagger|$, so that there are more values of the instrument than there are targeted treatments. We further argue that \eqref{eq:uniform} does not hold for some treatments unless $|\mathcal{D}|=|\mathcal{Z}|=2$. Such models therefore provide another class of ARUMs, distinct from the one with uniform targeting described in Example \ref{ex:arum}, for which Assumption \ref{as:encouragement} holds.  In Example \ref{ex:arum-violation} below, we show, however, that Assumption \ref{as:encouragement} does not hold when $|\mathcal D| \geq 3$, $|\mathcal{Z}| = |\mathcal{D}^\dagger|$, and the support of $(U_d: d \in \mathcal D)$ is $\Re^{|\mathcal D|}$.
\end{example}

\section{Examples of Models That Do Not Satisfy Assumption \ref{as:encouragement}} \label{sec:ex_not}

We now consider models that do not satisfy generalized monotonicity (Assumption \ref{as:encouragement}). For each model, we show that the identified sets for average potential outcomes are not given by \eqref{eq:bounds1}. For the first two examples, we further show that they do in fact provide identifying power beyond instrument exogeneity.

\begin{example} \label{eg:alt} 
Suppose $\mathcal{Y}=\mathcal{D}=\mathcal{Z}=\{0,1\}$. Let $\mathbf{Q}$ denote all distributions $Q$ that satisfy Assumption \ref{as:exog} and
\[ Q\{ D_0 = D_1\} =0~, \]
which, in the language of \cite{angrist1996identification}, is imposing that all individuals are either compliers or defiers.  For any $P$ such that $\mathbf{Q}_0(P, \mathbf{Q}) \neq \emptyset$, $\theta(Q)$ is identified relative to $\mathbf Q$, i.e., $\Theta_0(P, \mathbf Q)$ is a singleton. To see this, note that for any $Q\in \mathbf{Q}_0(P, \mathbf{Q})$ and $y \in \mathcal Y$,
\begin{align*}
   Q \{Y_1 = y\} & = Q \{Y_1 = y, D_0 = 0, D_1 = 1\} + Q \{Y_1 = y, D_0 = 1, D_1 = 0\} \\
   & = Q \{Y_1 = y, D_1 = 1\} + Q \{Y_1 = y, D_0 = 1\} \\
   & = Q \{Y_1 = y, D_1 = 1 \mid Z = 1\} + Q \{Y_1 = y, D_0 = 1 \mid Z = 0\} \\
   & = P \{Y = y, D = 1 \mid Z = 1\} + P \{Y = y, D = 1 \mid Z = 0\}~,
\end{align*}
where the first two equalities follows from $Q \{D_0 = D_1\} = 0$, the third equality follows from Assumption \ref{as:exog}, and the final equality follows from  $Q \in \mathbf{Q}_0(P, \mathbf{Q})$.   A similar argument establishes identification of $Q\{Y_0 = y\}$. In contrast, we show in Appendix B.2 that there exists a  $P$ for which $\mathbf Q_0(P,\mathbf Q) \neq \emptyset$ and \eqref{eq:bounds1} is not a singleton.  The identified set for $\theta(Q)$ is therefore not given by \eqref{eq:bounds1}.  We further show in Appendix B.2 that $\Theta_0(P, \mathbf Q_E^\ast)$ is not a singleton for the same $P$.  Thus, the model $\mathbf Q$ does have identifying power beyond instrument exogeneity for $\theta(Q)$.

Furthermore, recall as discussed in Remark \ref{remark:contradict} that if $\mathbf Q$ has identifying power for $\theta(Q)$, then it cannot contain a submodel satisfying instrument exogeneity and generalized monotonicity that is consistent with $P$. In this example, the model $\mathbf{Q}$ contains such a submodel if and only if $P$  satisfies either (i) $P\{D=0\mid Z=0\}=P\{D=1\mid Z=1\}=1$ (in which case all individuals are compliers) or  (ii) $P\{D=1\mid Z=0\}=P\{D=0\mid Z=1\}=1$ (in which case all individuals are defiers). In order for $\mathbf Q$ to have identifying power for $\theta(Q)$, it therefore must be the case that $0 < P\{D=0\mid Z=0\}, P\{D=0\mid Z=1\}< 1$, which is indeed satisfied by the counterexample in Appendix B.2. Further note that  \eqref{eq:bounds1} is a singleton in this example if and only if $P$ satisfies either (i) or (ii).
\end{example}

\begin{example} \label{ex:ordered}
Consider an ordered choice model for treatment. Suppose that $|\mathcal{Z}|\ge 3$, and let $\mathbf{Q}$ denote the set of all distributions that satisfy Assumption \ref{as:exog} and
\begin{equation} \label{eq:ordered}
    Q\{D_j \ge D_k\} =1 \text{ for all } j \ge k~.
\end{equation}
For example, $D$ might represent quantity of some treatment, and $Z$ might represent levels of subsidy for the treatment. The restriction in \eqref{eq:ordered} is equivalent to the monotonicity assumption considered in \cite{angrist1995two}. See \cite{vytlacil2006ordered} for the connection between this restriction and ordered discrete-choice selection models.
Without loss of generality, let $\mathcal D = \{0, \ldots, \bar D\}$. In this case, \eqref{eq:encouragement} is satisfied for  $d \in \{0,\bar{D}\}$ for all $Q \in \mathbf{Q}_0(P, \mathbf{Q})$; to see this, take $z^\ast(0)=\min\{\mathcal{Z}\}$ and $z^\ast(\bar{D})=\max\{\mathcal{Z}\}$. By a straightforward modification of the arguments underlying Theorem \ref{theorem:ev}, one can show that the identified sets for $\mathbb{E}_Q[Y_0]$ and $\mathbb{E}_Q[Y_{\bar{D}}]$ are given by \eqref{eq:bounds1} for any $P$ such that $\mathbf{Q}_0(P, \mathbf{Q}) \neq \emptyset$. The ordered monotonicity assumption in \eqref{eq:ordered} therefore has no identifying power beyond instrument exogeneity for $\mathbb{E}_Q[Y_0]$ and $\mathbb{E}_Q[Y_{\bar{D}}]$.  In contrast, \eqref{eq:encouragement} need not hold for $d \in \mathcal{D} \setminus \{0,\bar{D}\}$ and $Q \in \mathbf{Q}_0(P, \mathbf{Q})$.  In Appendix B.3, we show there exists a $P$ for which $\mathbf Q_0(P, \mathbf Q) \neq \emptyset$ and $\Theta_0(P, \mathbf Q)$ is not given by \eqref{eq:bounds1}. We further show $\Theta_0(P, \mathbf Q) \subsetneqq \Theta_0(P, \mathbf Q_E^\ast)$ for the same $P$. Thus, the ordered monotonicity assumption in \eqref{eq:ordered} does have identifying power beyond instrument exogeneity for $\mathbb{E}_Q[Y_d]$ for $d \in \mathcal{D} \setminus \{0,\bar{D}\}$. 
\end{example}

\begin{example} \label{ex:arum-violation}
In Example \ref{ex:lee2023}, we considered ARUMs satisfying the strict  one-to-one targeting assumption of \cite{lee2023treatment}. As discussed there, if $|\mathcal D| = 2$ or if $|\mathcal{D}| \ge 3$ and $|\mathcal{Z}| > |\mathcal{D}^\dagger|$, then Assumption \ref{as:encouragement} holds and the results in Section \ref{sec:main} are  applicable. Now consider the case in which $|\mathcal{D}| \ge 3$ and $|\mathcal{Z}| = |\mathcal{D}^\dagger|$. Denote by $\mathbf Q$ the ARUM model defined by \eqref{eq:ARUM} under the additional assumption that the support of $(U_d: d \in \mathcal D)$ is $\Re^{|\mathcal{D}|}$. Then, as we show in Appendix B.4, while  \eqref{eq:encouragement} will hold for targeted treatments,  \eqref{eq:encouragement} cannot hold for any non-targeted treatment, and thus Assumption \ref{as:encouragement} is violated. We show that, while the identified set for $\mathbb{E}_Q[Y_d]$ is given by \eqref{eq:bounds1} for targeted treatments when  $\mathbf Q_0(P, \mathbf Q) \neq \emptyset$, there exists a $P$ for which $\mathbf Q_0(P, \mathbf Q) \neq \emptyset$ and the identified set for $\mathbb{E}_Q[Y_d]$ is not given by  \eqref{eq:bounds1} for non-targeted treatments.
\end{example}

\section*{Disclosure Statement}

The authors report there are no competing interests to declare.

\newpage
\appendix
\section{Proofs of Main Results}

\subsection{Proof of Lemma \ref{lem:maxp}}
The necessity of \eqref{eq:zd-id} has been proved in the main text. On the other hand, fix a $d \in \mathcal D$, suppose \eqref{eq:zd-id} holds for some $z \in \mathcal Z$. Fix a particular $z^\ast \in \mathcal Z^\ast(d, Q)$ that satisfies \eqref{eq:encouragement}. If $z = z^\ast$ then of course Assumption \ref{as:encouragement} holds for $z$. Suppose $z \neq z^\ast$. Then, 
\begin{align*}
0 \le & ~ P\{D=d \mid Z=z\} - P\{D=d \mid Z=z^\ast\}    \\
=& ~ Q\{D_z=d\} - Q\{D_{z^\ast}=d\} \\
=& ~ Q\{D_{z^\ast} \neq d, D_z = d\} - Q\{D_z \ne  d, D_{z^\ast} = d \}
 \\
=& ~ - Q\{D_z \ne  d, D_{z^\ast} = d \}~,
\end{align*}
where the first inequality is using that $z$ satisfies \eqref{eq:zd-id}, the second line is using Assumption \ref{as:exog}, and the last line is using that $z^\ast$ satisfies \eqref{eq:encouragement}.  Thus $Q\{D_z \ne  d, D_{z^\ast} = d \} =0$.  Assumption \ref{as:encouragement} implies 
\begin{align*}
    &Q\{D_z \neq d, D_{z'} = d \text{ for some } z' \in \mathcal Z\} \\ 
    &= Q \{  D_z \neq d, D_{z^\ast} = d\} + Q\{ D_z \neq d, D_{z^\ast} \neq d, D_{z'} = d ~\text{for some}~ z' \in \mathcal Z\setminus\{z^\ast\} \} \\
    &= 0 ~,
\end{align*}
and $z$ satisfies Assumption \ref{as:encouragement} as well. \qed

\subsection{Proof of Theorem \ref{theorem:idpower}}
The desired result follows immediately from Theorems \ref{theorem:ev} and \ref{theorem:same}. \qed

\subsection{Proof of Theorem \ref{theorem:ev}}

This section is organized as follows. In Section \ref{sec:auxillary}, we introduce additional notation that is helpful in formally proving our result, including defining subgroups of individuals, called treatment response types, who are defined by what treatment they would take at each value of the instrument. Because all variables are discrete, we will directly work with the probability mass function. We derive a lemma that characterizes a sufficient condition for a given distribution of potential outcomes and treatments $Q$ to rationalize the distribution of the data $P$. The lemma states that whether a $Q$ that satisfies Assumption \ref{as:exog} rationalizes $P$ depends on the probability of each treatment response type and on the probability of that type's potential outcomes corresponding to treatments they would choose for some value of the instrument (so that they ``comply with'' this treatment at least for some values of the instrument), but does not depend on the probability of that type's potential outcomes corresponding to treatments they would not choose for any value of the instrument (so that they are ``never-takers'' of this treatment). If $Q$ satisfies Assumptions \ref{as:exog} and \ref{as:encouragement}, then the set of treatments for which a treatment response type is a ``never-taker'' are precisely the set of treatments that they would not take even when maximally encouraged to do so. Therefore, the implication of the lemma is that if $Q$ satisfies Assumptions \ref{as:exog} and \ref{as:encouragement} and rationalizes $P$, then any other $Q^\ast $ satisfying Assumption  \ref{as:exog} and \ref{as:encouragement} will also  rationalize $P$ if $Q$ and $Q^\ast$ differ \emph{only} in the probability of potential outcomes corresponding to treatments that a given response type would not take even when maximally encouraged to do so.

In Section \ref{sec:theorem:ev}, we use the notation and lemma introduced in Section  \ref{sec:auxillary} to prove the theorem.  Let $\mathbf{Q} $  satisfy the assumptions of the theorem.
We first show that $\Theta_0(P, \mathbf{Q})$ is a subset of the bounds in \eqref{eq:bounds1}. We then show that the bounds in \eqref{eq:bounds1} are a subset of $\Theta_0(P, \mathbf{Q})$ using the following proof strategy.  By assumption, $\Theta_0(P, \mathbf{Q})$ is non-empty, so that there exists  a distribution $Q \in \mathbf{Q}$ that rationalizes $P$.  For each value $\theta_0$ in the bounds of \eqref{eq:bounds1}, we construct an alternative distribution $Q^\ast \in \mathbf{Q}$ such that $\theta(Q^\ast) = \theta_0$  with $Q$ and $Q^\ast$ differing only in the probability of outcomes corresponding to treatments that a given response type would not take even when maximally encouraged to do so.  That the constructed $  Q^\ast$ lies in $ \mathbf{Q}$ follows from the assumption that $\mathbf Q$ satisfies Assumption \ref{as:norestrictY} and that $Q$ and $Q^\ast$ have the same distribution of potential treatment choices with $Q \in \mathbf{Q}$.  That the  constructed $Q^\ast$ rationalizes $P$ follows from  $Q \in \mathbf{Q}$ and the previously described lemma.  That we are able to construct such a $Q^\ast$ with $\theta(Q^\ast) = \theta_0$ for every $\theta_0$ in the bounds of  \eqref{eq:bounds1} establishes that the bounds \eqref{eq:bounds1} are a subset of $\Theta_0(P, \mathbf{Q})$, completing the proof.



\subsubsection{Auxillary Results} \label{sec:auxillary}
To present the proof of Theorem \ref{theorem:ev}, we first introduce some further notation. Because all variables are discrete, we will directly work with the probability mass function. Recall from the discussion in Section \ref{sec:setup} that if $Q$ satisfies Assumption \ref{as:exog}, then $P = QT^{-1}$ if and only if $$p_{yd|z} = Q \{Y_d = y, D_z = d\}~.$$  Following \cite{heckman2018unordered}, we define a treatment response
type as a vector $r^t \in   \mathcal D^{|\mathcal Z|}$, 
\begin{align*}
r^t & = (d_z : z \in \mathcal Z) \in  \mathcal D^{|\mathcal Z|}~.
\end{align*}
Treatment response types are also called principal strata \citep{frangakis2002principal}.  We 
analogously define an outcome response type as a vector $r^o \in   \mathcal Y^{|\mathcal D|}$, 
\begin{align*}
r^o & = (y_d : d \in \mathcal D) \in   \mathcal Y^{|\mathcal D|}~. 
\end{align*}
Because all variables are discrete, we define the probability mass function as
\[ q \big ( r^o, r^t \big ) = Q\{(Y_d : d \in \mathcal D) = r^o, (D_z : z \in \mathcal Z) =  r^t\}~. \]
For the rest of the proof, without loss of generality, we suppose $\mathcal D = \{0, 1, \dots, |\mathcal D| - 1\}$ and $\mathcal Z = \{0, 1, \dots, |\mathcal Z| - 1 \}$. Let $r_j^o$ denote the $(j + 1)$th entry of $r^o$ and $r_j^t$ denote the $(j + 1)$th entry of $r^t$. In other words, $r_j^o$ denotes the value of the potential outcome for the outcome response type $r^o$ when the treatment equals $j$, and $r_j^t$ denotes the value of the potential treatment for the treatment response type $r^t$ when the instrument equals $j$. In this notation, if $Q$ satisfies Assumption \ref{as:exog}, then it follows from \eqref{eq:p-q} that $P = QT^{-1}$ if and only if 
\begin{equation} \label{eq:p-q-ev}
p_{yd|z}=  \sum_{(r^o,r^t) : r^o_d = y, r^t_z = d} q(r^o, r^t) \qquad \forall ~ y \in \mathcal Y, ~ d \in \mathcal{D}, ~ z \in \mathcal{Z}~.
\end{equation}

Below we derive a lemma that simplifies determining whether $q(r^o, r^t)$ satisfies \eqref{eq:p-q-ev} and will be used subsequently to derive our characterization of the identified set.  To this end, we require some further notation. Let 
\begin{align*}
    \mathcal{N}(r^t) & = \{ d \in \mathcal{D} ~:~ r^t_z  \ne d ~ \mbox{for all} ~ z \in \mathcal{Z}\}~, \\
    \mathcal{N}(r^t)^c & = \{ d \in \mathcal{D} ~:~ r^t_z  = d ~ \mbox{for some} ~ z \in \mathcal{Z}\}~,
\end{align*}
For a given treatment response type $r^t$, $\mathcal{N}(r^t)$ is the set of treatments for which that treatment response type is a ``never-taker,'' and $\mathcal{N}(r^t)^c$ is the set of treatments for which that treatment response type will ``comply with'' the treatment for some value of $z$. Using this notation, partition outcome and treatment response types $(r^o, r^t)$ as $(r^o_n(r^t), r^o_c(r^t), r^t)$ where
\begin{align*} r^o_n(r^t) &=  (  r^o_d :  d \in \mathcal{N}(r^t) )~,\\
r^o_c(r^t) &=  (  r^o_d :  d  \in \mathcal{N}(r^t)^c )~. \end{align*}
For a given treatment response type $r^t$, $r^o_n(r^t)$ are those outcomes that are never observed for that response type, and $ r^o_c(r^t)$ are the remaining outcomes that are observed given some value of $Z$. Here, the subscripts $n$ and $c$ stand for ``never-taker'' and ``complier.''

\begin{remark}
Here we illustrate how our notation specializes under Assumption \ref{as:encouragement}. Note Assumption \ref{as:encouragement} can be expressed as restricting $q(r^o, r^t) = 0$ unless the treatment response type $r^t$ satisfies the condition therein; in other words, it restricts the support of the treatment response type. In particular, if for some $d \in \mathcal{D}$, $r^t_{z^\ast(d)} \ne d$ while $r^t_{z^{\prime}} =d$ for some $z' \neq z^\ast(d)$, then $q(r^o,r^t)=0$ for all $r^o$. For any $r^t$ in the support,
\begin{align*} 
 \mathcal{N}(r^t) &= \{ d \in \mathcal{D} ~:~ r^t_{z^\ast(d)} \ne d \}~,\\ \mathcal{N}(r^t)^c &= \{ d \in \mathcal{D} ~:~ r^t_{z^\ast(d)} = d \}~, 
\intertext{and}
    r^o_n(r^t) &= (  r^o_{d} :  d\in \mathcal{D}, ~ r^t_{z^\ast(d)} \ne d )~,\\
    r^o_c(r^t) &= (   r^o_{d} :  d\in \mathcal{D}, ~ r^t_{z^\ast(d)} = d )~. 
\end{align*}
Indeed, $z^\ast(d)$ is the instrument that maximally encourages to treatment $d$, so if $r^t_{z^\ast(d)} \ne d$, then $r^t_z \ne d$ for all $z \in \mathcal Z$. As a result, someone with that treatment response type $r^t$ never takes $d$, and hence $d \in \mathcal N(r^t)$. Otherwise, $d \in \mathcal N(r^t)^c$, or this person is a ``complier'' for treatment $d$ at least when $z = z^\ast(d)$. The outcome response type $r^o$ is then partitioned into $r^o_n(r^t)$ and $r^o_c(r^t)$ according to whether $d \in \mathcal N(r^t)$ or not.
\end{remark}

For notational convenience, we further define the probability mass $q(r^o_c(r^t), r^t)$ as $q(r^o_n(r^t), \allowbreak r^o_c(r^t), r^t)$ summed over
$r^o_n(r^t)$:
$$q( r^o_c(r^t), r^t) = \begin{cases} q( r^o, r^t) & \mbox{if} ~  \mathcal{N}(r^t) = \emptyset ~ \mbox{so that}~
r^o_c(r^t) = r^o  \\
\sum_{r^o_n(r^t) \in \mathcal{Y}^{|\mathcal N(r^t)|}}  
    q(r^o_n(r^t), r^o_c(r^t), r^t) & \mbox{if}  ~    \mathcal{N}(r^t) \ne \emptyset
     ~ \mbox{so that}~
r^o_c(r^t)  \ne r^o~.
\end{cases}$$
In defining the probability mass $q(r^o_c(r^t), r^t)$, we sum over all possible values of $r^o_n(r^t)$, because these are the outcomes of treatments that are never taken by the treatment response type $r^t$, and hence will not be relevant for the observed data. Using this notation, we have the following lemma that asserts whether $q( r^o, r^t)$ satisfies \eqref{eq:p-q-ev} depends only on $ q( r^o_c(r^t), r^t)$. This lemma implies that whether a distribution of potential outcomes and treatments $Q$ that satisfies Assumption \ref{as:exog} rationalizes the distribution of the data $P$ depends only on the probability of each treatment response type and the probability of that type's potential outcomes that would be observed for some value of the instrument.

\begin{lemma}\label{lemma:ev1}
    Suppose $q$ satisfies \eqref{eq:p-q-ev}.  Then, $q^\ast$ satisfies
    \eqref{eq:p-q-ev} if, for each $ r^t \in \mathcal D^{|\mathcal Z|} $,
    \begin{equation} \label{eq:lemma}
    q^\ast(r^o_c(r^t), r^t)=   q(  r^o_c(r^t), r^t) ~~ \forall ~ r^o_c(r^t)~.
    \end{equation}
\end{lemma}

\begin{proof}
We can rewrite \eqref{eq:p-q-ev} as
\begin{align*}
    p_{yd|z}=     &~  \sum_{r^t : r^t_{z}=d } ~~  \sum_{r^o : r^o_d = y}   q(r^o, r^t) \\
   =&       \sum_{r^t : r^t_{z}=d }   ~~  \sum_{r^o_c(r^t): r^o_d = y } ~~ 
   \left( \sum_{r^o_n(r^t)} q(r^o_n(r^t),r^o_c(r_t), r^t) \right)\\
  =&      \sum_{r^t : r^t_{z}=d }   ~~  \sum_{r^o_c(r^t): r^o_d = y } ~~  q(r^o_c(r_t), r^t) ~,
\end{align*}
where the second equality uses that $r_c^o(r^t)$ is nonempty because $r^t_{z} = d$ and that $r^o_d$ is an element of $r^o_c(r^t) $ for $r^t$ such that $r^t_{z} = d$. The result now follows.
\end{proof}

\subsubsection{Proof of the Theorem} \label{sec:theorem:ev}
\underline{$\Theta_0(P,\mathbf{Q}) \subseteq \eqref{eq:bounds1}$}

We first show that \eqref{eq:bounds1} provides valid bounds on $\theta(Q)$ under the stated assumptions, that is, $\Theta_0(P,\mathbf{Q})$ is a subset of the bounds of \eqref{eq:bounds1}.
Suppose $Q \in \mathbf{Q}_0(P,\mathbf{Q})$. For each $d \in \mathcal{D}$, 
\begin{align*}  \mathbb{E}_Q[Y_d] &= \mathbb{E}_Q[Y_d \mathbbm 1\{ D_{z^\ast(d)}  = d\}  ] + \mathbb{E}_Q[Y_d \mathbbm 1\{ D_{z^\ast(d)} \ne d\}]\\
  &= \beta_{d|z^\ast(d)} + \mathbb{E}_Q[Y_d \mathbbm 1\{ D_{z^\ast(d)} \ne d\} ]~,
\end{align*}
where the second equality is using Assumption \ref{as:exog}.  We have 
$$ \mathbb{E}[Y_d \mathbbm 1\{ D_{z^\ast(d)} \ne d\}] \in [y^L Q\{D_{z^\ast(d)} \ne d\}, ~y^U Q\{D_{z^\ast(d)} \ne d\} ] ~,$$
while Assumption \ref{as:exog} implies that   $$Q\{D_{z^\ast(d)} \ne d\}= 1 - \sum_{y \in \mathcal{Y}} p_{yd|z^\ast(d)}~.$$  We thus have that 
$$ \mathbb{E}[Y_d] \in [ \beta_{d|z^\ast(d)} +y^L (1 - \sum_{y \in \mathcal{Y}} p_{yd|z^\ast(d)})  ,~ \beta_{d|z^\ast(d)}  +y^U (1 - \sum_{y \in \mathcal{Y}} p_{yd|z^\ast(d)}) ] ~,$$ for each $d \in \mathcal{D}$, and thus \eqref{eq:bounds1} provides valid bounds on $\theta(Q)$ under the stated assumptions.

\noindent \underline{$\eqref{eq:bounds1} \subseteq \Theta_0(P,\mathbf{Q})$}

We now show that the bounds of \eqref{eq:bounds1} are the identified set for $\theta(Q)$, that is, the bounds of \eqref{eq:bounds1} are a subset of $\Theta_0(P, \mathbf{Q})$. 
Let $q$ denote latent probabilities corresponding to a fixed $Q \in \mathbf{Q}_0(P, \mathbf{Q})$.   There exists such a $q$ by the assumption that $\mathbf{Q}_0(P, \mathbf{Q})$ is non-empty.
We show that for each $\theta_0$ in the right-hand side of \eqref{eq:bounds1}, we can construct an alternative distribution of potential outcomes and treatments $Q^\ast$ that is contained in $\mathbf{Q}_0(P, \mathbf{Q})$ and for which $\theta(Q^{\ast}) = ( \mathbb{E}_{Q^{\ast}}[Y_j] : j \in \mathcal{D} ) $ is equal to $\theta_0$. In particular, for each $\theta_0$ in the right-hand side of \eqref{eq:bounds1} we will construct $q^\ast$ corresponding to $Q^\ast$ that
\begin{enumerate}[\rm (a)]
    \item satisfies $\sum_{r^o} q^\ast(r^o,r^t)=\sum_{r^o} q(r^o,r^t)$ and hence $Q^\ast \in \mathbf{Q}$ because the distribution of $(D_z: z \in \mathcal Z)$ is unchanged and by assumption that $\mathbf Q$ satisfies Assumption \ref{as:norestrictY},
    \item satisfies \eqref{eq:lemma} and hence $P = Q^{\ast} T^{-1}$ due to Lemma \ref{lemma:ev1}, and
    \item satisfies $\theta({Q^\ast}) = \theta_0$.
\end{enumerate}
Properties (a) and (b) allow us to conclude that $Q^{\ast} \in \mathbf{Q}_0(P, \mathbf{Q})$, that is, the constructed distribution is consistent with $P$ and the model $\mathbf{Q}$. These properties will follow from our iterative construction of $q^\ast$, which preserves the marginal distribution of potential treatments but modifies the marginal distributions of potential outcomes for outcomes that are never observed for a given treatment response type, that is, correspond to a never-taken treatment for a given treatment response type. Because the marginal distribution of potential treatments is preserved, property (a) follows. Because only the marginal distributions of potential outcomes for never-taken treatments are modified, property (b) follows. Property (c) follows from being able to flexibly modify the marginal distributions of potential outcomes for never-taken treatments, so that any $\theta$ can be achieved.


\noindent \underline{Part 1: construct $Q^\ast$}

We now construct an alternative $q^\ast$ as follows.
Fix some vector of weights $\alpha = (\alpha_0,\alpha_1,...,\alpha_{|\mathcal D|-1})^{\prime} \in [0,1]^{|\mathcal D|}$ to be specified below.
For each treatment response type $r^t$, let $q^\ast_0(r^o,r^t)= q(r^o,r^t)$ for all $r^o$. Let $K(r^t) = |\mathcal{N}(r^t)|$ be the number of treatments for which treatment response type $r^t$ is a never-taker. Note that if $K(r^t) = 0$, then $\mathcal N(r^t) = \emptyset$, so $r^t_{z^\ast(d)} = d$ for all $d \in \mathcal{D}$ under $r^t$. For such an $r^t$, we set $q^\ast(r^o, r^t) = q^\ast_0(r^o,r^t) = q(r^o, r^t)$ for all $r^o \in \mathcal{Y}^{|\mathcal D|}$. 

If $K(r^t) \ge 1$, enumerate the set of never-taken treatments $\mathcal{N}(r^t)$ as $\{ j[1],...,j[K(r^t)]\}$,   and  for $k=1$ to $K(r^t)$, define $q_k^\ast$ iteratively as follows:
\begin{equation} \label{eq:silly}
    \begin{split}
        q^\ast_{k}((r^o_{-j[k]}, r^o_{j[k]}=y^L),r^t) &= (1-\alpha_{j[k]})  \sum_{r^o_{j[k]} \in \mathcal{Y}} q^\ast_{k-1}((r^o_{-j[k]}, r^o_{j[k]}), r^t) \\
        q^\ast_{k}((r^o_{-j[k]}, r^o_{j[k]}=y),r^t) &= 0 \qquad \mbox{for}~ y \not \in \{y^L, y^U\}  \\
        q^\ast_{k}((r^o_{-j[k]}, r^o_{j[k]}=y^U),r^t) &= \alpha_{j[k]} \sum_{r^o_{j[k]} \in \mathcal{Y}} q^\ast_{k-1}((r^o_{-j[k]}, r^o_{j[k]}), r^t)~,
\end{split}
\end{equation}
for all $r^o_{-j[k]}$, where we partition $r^o = (r^o_{-j[k]}, r^o_{j[k]})$. 
Intuitively, in step $k$, for each $r^o_{-j[k]}$ and $r^t$ we reassign the probabilities of all outcome responses to never-taken treatment $j[k]$ to outcome responses $y^L$ and $y^U$, splitting between $y^L$ and $y^U$ according to weight $\alpha_{j[k]}$.

With this construction, the marginal distribution of $Y_{j[k]}$ for treatment response type $r^t$ is only modified in step $k$. This statement implies that for each fixed $k$, for step $\ell \leq k - 1$, and for any outcome $y \in \mathcal{Y}$,
\begin{equation} \label{eq:<k}
    \sum_{r_{-j[k]}^o} q^\ast_{\ell}((r^o_{-j[k]}, r^o_{j[k]}=y), r^t) = \sum_{r_{-j[k]}^o} q^\ast_0((r^o_{-j[k]}, r^o_{j[k]}=y),r^t) = \sum_{r_{-j[k]}^o} q((r^o_{-j[k]}, r^o_{j[k]}=y),r^t)~.
\end{equation}
%
On the other hand, for each fixed $k$, because the marginal distribution of $Y_{j[k]}$ for treatment response type $r^t$ is not further modified after step $k$, \eqref{eq:silly} and \eqref{eq:<k} imply
\begin{align} 
\nonumber    \sum_{r_{-j[k]}^o} q^\ast_{K(r^t)}((r^o_{-j[k]}, r^o_{j[k]}=y^L),r^t) & = \sum_{r_{-j[k]}^o} q^\ast_{k}((r^o_{-j[k]}, r^o_{j[k]}=y^L),r^t)\\
\nonumber    & = (1 - \alpha_{j[k]}) \sum_{r^o_{j[k]} \in \mathcal{Y}} \sum_{r_{-j[k]}^o} q((r^o_{-j[k]}, r^o_{j[k]}),r^t)\\
\nonumber            & = (1 - \alpha_{j[k]}) \sum_{r^o} q(r^o, r^t) ~,\\
\nonumber    \sum_{r_{-j[k]}^o} q^\ast_{K(r^t)}((r^o_{-j[k]}, r^o_{j[k]}=y),r^t) & = 0  \qquad \mbox{for}~ y \not \in \{y^L, y^U\}  ~,\\
\nonumber    \sum_{r_{-j[k]}^o} q^\ast_{K(r^t)}((r^o_{-j[k]}, r^o_{j[k]}=y^U),r^t) & = \sum_{r_{-j[k]}^o} q^\ast_{k}((r^o_{-j[k]}, r^o_{j[k]}=y^U),r^t) \\
\nonumber    & = \alpha_{j[k]} \sum_{r^o_{j[k]} \in \mathcal{Y}} \sum_{r_{-j[k]}^o} q((r^o_{-j[k]}, r^o_{j[k]}),r^t) \\
\label{eq:marginal}    & = \alpha_{j[k]} \sum_{r^o} q(r^o, r^t) ~.
\end{align}
These equations state that under $q^\ast_{K(r^t)}$, the constructed distribution in the final step $K(r^t)$, the probability that $r^o_{j[k]} = y$ for each $r^t$ is zero if $y$ is not $y^L$ or $y^U$, is $\alpha_{j[k]}$ times the true probability of $r^t$ under $q$ if $y = y^U$, and is $1-\alpha_{j[k]}$ times the true probability of $r^t$ under $q$ if $y = y^L$.

Finally, set
$$q^\ast(r^o,r^t)= q^\ast_{K(r^t)}(r^o,r^t)  \qquad   \forall ~ r^o~.$$

\noindent \underline{Part 2: verify property (a), $Q^\ast \in \mathbf Q$}

With this construction, $q^\ast(r^o,r^t)$ is non-negative and from \eqref{eq:marginal} we have
$$\sum_{r^o} q^\ast(r^o,r^t) =\sum_{r^o} q(r^o,r^t) ~~ \mbox{for all} ~~ r^t~.$$ 
Because we assume $\mathbf Q$ satisfies Assumption \ref{as:norestrictY} and $Q \in \mathbf{Q}$, this implies that $Q^\ast \in \mathbf{Q}$.

\noindent \underline{Part 3: verify property (b), $P = Q^\ast T^{-1}$}

Furthermore, for each $r^t$ and for all $r^o_c(r^t)$, the construction of $q^\ast$ in \eqref{eq:silly} implies
\begin{align*}
    q^\ast(r^o_c(r^t), r^t) &= \sum_{r^o_n(r^t) \in \mathcal{Y}^{\vert \mathcal{N}(r^t) \vert}} q^\ast(r^o_n(r^t), r^o_c(r^t), r^t) \\
    &= \sum_{r^o_{j[1]} \in \mathcal{Y}} \cdots \sum_{r^o_{j[K(r^t)]} \in \mathcal{Y}} q_{K(r^t)}^\ast((r^o_{j[1]}, \dots, r^o_{j[K(r^t)]}), r^o_c(r^t), r^t) \\
    &= \sum_{r^o_{j[1]} \in \mathcal{Y}} \cdots \sum_{r^o_{j[K(r^t)]} \in \{y^L, y^U\}} q_{K(r^t)}^\ast(r^o_{-j[K(r^t)]}, r^o_{j[K(r^t)]}, r^t) \\
    &= \sum_{r^o_{j[1]} \in \mathcal{Y}} \cdots \sum_{r^o_{j[K(r^t)]} \in \mathcal{Y}} q_{K(r^t)-1}^\ast(r^o_{-j[K(r^t)]}, r^o_{j[K(r^t)]}, r^t) \\
    &= \sum_{r^o_{j[1]} \in \mathcal{Y}} \cdots \sum_{r^o_{j[K(r^t)]} \in \mathcal{Y}} q_0^\ast((r^o_{j[1]}, \dots, r^o_{j[K(r^t)]}), r^o_c(r^t), r^t) \\
    &= \sum_{r^o_n(r^t) \in \mathcal{Y}^{\vert \mathcal{N}(r^t) \vert}} q(r^o_n(r^t), r^o_c(r^t), r^t) \\
    &= q(r^o_c(r^t), r^t)~.
\end{align*}
Therefore for each $r^t$, $q^\ast(r^o_c(r^t), r^t)=   q(  r^o_c(r^t), r^t) ~~ \forall ~ r^o_c(r^t)$, and hence by Lemma \ref{lemma:ev1}, $q^\ast$ satisfies \eqref{eq:p-q-ev} so that $P = Q^\ast T^{-1}$.
Thus $Q^\ast \in \mathbf{Q}_0(P, \mathbf{Q})$.

\noindent \underline{Part 4: verify property (c), $\theta(Q^\ast) = \theta_0$}

Note for each $d \in \mathcal{D}$,
$$\mathbb{E}_{Q^\ast}[Y_d \mathbbm 1\{ D_{z^\ast(d)}  = d\}  ]  = \mathbb{E}_P[ Y \mathbbm{1}\{D=d\} \mid Z=z^\ast(d)] = \beta_{d|z^\ast(d)}~.$$
Further note since Assumption \ref{as:encouragement} holds for $Q$, we have that for each $d \in \mathcal{D}$ and for each $r^t$ if $r^t_{z^\ast(d)}=d$ then $r^o_d$ is a component of $r_c^o(r^t)$ and $d \in \mathcal{N}(r^t)^c$, while if $r^t_{z^\ast(d)} \ne d$ then $r^o_d$ is a component of $r_n^o(r^t)$ and $d \in \mathcal{N}(r^t)$.
Then we also have that for each $d \in \mathcal{D}$,
\begin{align*} 
\mathbb{E}_{Q^{\ast}}[Y_d \mathbbm 1\{ D_{z^\ast(d)} \ne d\} ] &=
 \sum_{y \in \mathcal{Y}} ~~ \sum_{r^t : r^t_{z^\ast(d)} \ne d} ~~ \sum_{r^o : r^o_d =y} ~~ y ~q^\ast(r^o,r^t)\\
&= \sum_{y \in \mathcal{Y}} ~~ \sum_{r^t : r^t_{z^\ast(d)} \ne d} ~~ \sum_{r^o_{-d}} ~~ y ~q^\ast((r^o_{-d}, r^o_d=y),r^t)\\
&=  (\alpha_d y^U + (1-\alpha_d) y^L)   \sum_{r^t : r^t_{z^\ast(d)} \ne d} ~~ \sum_{r^o }    q(r^o,r^t)   \\
&= (\alpha_d y^U + (1-\alpha_d) y^L)  ~ Q\{D_{z^\ast(d)} \ne d\}\\
&=  (\alpha_d y^U + (1-\alpha_d) y^L)   (1 - \sum_y p_{yd|z^\ast(d)})~,
\end{align*} 
where the third equality is using that \eqref{eq:marginal} holds for $r^t$ such that $r^t_{z^\ast(d)}\ne d$, so that $d = j[k']$ for some $k'$ in constructing $q^\ast(\cdot, r^t)$ for that $r^t$, and the last equality is using that $Q$ satisfies \eqref{eq:p-q-ev}. Thus, for each $d \in \mathcal{D}$,
\begin{align*} \mathbb{E}_{Q^\ast}[Y_d] &= \mathbb{E}_{Q^\ast}[Y_d \mathbbm 1\{ D_{z^\ast(d)}  = d\}  ] + \mathbb{E}_{Q^\ast}[Y_d \mathbbm 1\{ D_{z^\ast(d)} \ne d\}]\\
  &= \beta_{d|z^\ast(d)} + (\alpha_d y^U + (1-\alpha_d) y^L)   (1 - \sum_y p_{yd|z^\ast(d)})~.
\end{align*}
For any $\theta_0$ contained in \eqref{eq:bounds1}, we can thus choose  $\alpha = (\alpha_0,\alpha_1,...,\alpha_{|\mathcal D|-1})^{\prime} \in [0,1]^{|\mathcal D|}$ such that $\theta({Q^\ast}) = \theta_0$. \qed

\subsection{Proof of Corollary \ref{cor:ev2}}
The results follows immediately from Theorem \ref{theorem:ev} because $\mathbb{E}[Y_j  ] - \mathbb{E}[Y_k]$ is a function of $\theta(Q)$. \qed

\subsection{Proof of Lemma \ref{lem:mean_iv}}
To see it, note for any $d \in \mathcal D$ and $z \in \mathcal Z$,
\begin{align*}
    \mathbb E_Q[Y_d] = \mathbb E_Q[Y_d \mid Z = z] & = \mathbb E_Q[Y_d \mathbbm 1 \{D = d\} \mid Z = z] + \mathbb E_Q[Y_d \mathbbm 1 \{D \neq d\} \mid Z = z] \\
    & = \mathbb E_Q[Y \mathbbm 1 \{D = d\} \mid Z = z] + \mathbb E_Q[Y_d \mathbbm 1 \{D \neq d\} \mid Z = z] \\
    & = \beta_{d | z} + \mathbb E_Q[Y_d \mathbbm 1 \{D \neq d\} \mid Z = z] \\
    & \leq \beta_{d | z} + y^U P \{D \neq d \mid Z = z\} \\
    & = \beta_{d | z} + y^U(1 - \sum_{y \in \mathcal Y} p_{yd|z})~.
\end{align*}
Because the inequality holds for all $z \in \mathcal Z$, the upper end for each $d \in \mathcal D$ of \eqref{eq:meaniv2} is a valid upper bound for $\mathbb E[Y_d]$. On the other hand, they can be simultaneously attained for all $d \in \mathcal D$ by setting $Y_d = y^U$ whenever $D \neq d$ and $Z = z$, without affecting the distribution of $(Y, D, Z)$. A similar argument can be applied to the lower ends. In addition, any values in between can also be attained simultaneously for all $d \in \mathcal D$ by setting $Y_d$ to be a convex combination of $y^L$ and $y^U$ whenever $D \neq d$ and $Z = z$ without affecting the distribution of $(Y, D, Z)$, and therefore \eqref{eq:meaniv2} is indeed the identified set for $\theta(Q)$ under mean independence. \qed

\subsection{Proof of Lemma \ref{lemma:ev2}}

Suppose $Q \in \mathbf{Q}_0(P, \mathbf{Q}) \ne \emptyset$.
Consider the upper endpoints of \eqref{eq:bounds1}  and \eqref{eq:meaniv2}.  
For any $d \in \mathcal{D}$, $z \in \mathcal{Z}$,   
\begin{align*}
\mathbb{E}_P[Y &\mathbbm{1}\{D=d\} ]\mid Z = z^\ast(d)]  + y^U \mathbb{E}_P[ 1- \mathbbm{1}\{D=d\} \mid Z = z^\ast(d)]\\  & ~~~~ -
\mathbb{E}_P[Y \mathbbm{1}\{D=d\} \mid Z = z] - y^U \mathbb{E}_P[ 1- \mathbbm{1}\{D=d\} \mid Z = z] \\
&= \mathbb{E}_Q[Y_d \mathbbm{1}\{D_{z^\ast(d)}=d\}]+
 y^U \mathbb{E}_Q[ 1- \mathbbm{1}\{D_{z^\ast(d)}=d\}  ] \\
& \hspace{3em} -
\mathbb{E}_Q[Y_d \mathbbm{1}\{D_{z}=d\}] -
 y^U \mathbb{E}_Q[ 1- \mathbbm{1}\{D_{z}=d\}  ]\\
&= \mathbb{E}_Q[( Y_d-y^U) ( \mathbbm{1}\{D_{z^\ast(d)}=d\} - \mathbbm{1}\{D_{z}=d\})]   \\
&= \mathbb{E}_Q[ ( Y_d - y^U)  ( \mathbbm{1}\{D_{z^\ast(d)}=d, D_{z} \ne d\} - \mathbbm{1}\{D_{z^\ast(d)}\ne d, D_{z} = d\})]\\
&= \mathbb{E}_Q[ ( Y_d - y^U)   \mathbbm{1}\{D_{z^\ast(d)}=d, D_{z} \ne d\}]\\
 & \le 0~,
\end{align*}
where the first equality uses Assumption \ref{as:exog} and the fourth equality uses that $Q\{ D_{z^\ast(d)}\ne d,  D_z= d\}=0$ for all $Q$ satisfying Assumption \ref{as:encouragement}. Since this inequality holds for all  $z \in \mathcal{Z}$, we have that the upper endpoint of the interval in \eqref{eq:bounds1} is weakly smaller than the upper endpoint of the interval in \eqref{eq:meaniv2}.  Conversely, the upper endpoint of \eqref{eq:bounds1} is contained  in the set over which the upper endpoint of  \eqref{eq:meaniv2} is minimizing over, and thus the  upper endpoint of \eqref{eq:bounds1} is weakly larger than the upper endpoint of  \eqref{eq:meaniv2}.  We conclude that the upper endpoints are the same. Parallel arguments show the equivalence of the lower endpoints. \qed

\subsection{Proof of Theorem \ref{theorem:same}}
Because by assumption $\mathbf Q \subseteq \mathbf Q' \subseteq \mathbf Q_{\it MI}^\ast$, we have
$$\Theta_0(P,\mathbf{Q}) \subseteq \Theta_0(P, \mathbf{Q}') \subseteq \Theta_0(P, \mathbf{Q}_{\it MI}^\ast)~.$$
By Lemma \ref{lemma:ev2}, if $\mathbf Q \subseteq \mathbf Q_{E, M}^\ast$, $\mathbf Q$ satisfies Assumption \ref{as:norestrictY}, and $\mathbf{Q}_0(P, \mathbf{Q}) \ne \emptyset$, then we have
$$\Theta_0(P,\mathbf{Q}) = \Theta_0(P, \mathbf{Q}_{\it MI}^\ast)~.$$
The result now follows by a sandwich argument. \qed

\subsection{Proof of Corollary \ref{cor:exog}}
The result follows by taking $\mathbf Q = \mathbf Q_{E, M}^\ast$ and $\mathbf Q' = \mathbf Q_E^\ast$, noting that $\mathbf Q_E^\ast \subseteq \mathbf Q_{\it MI}^\ast$ because Assumption \ref{as:exog} implies Assumption \ref{as:mean_iv}. \qed

\subsection{Proof of Corollary \ref{cor:marginal}}
To begin, note by assumption $\mathbf Q_0(P, \mathbf Q_{\it WE, M}^\ast) \neq \emptyset$, so there exists $Q \in \mathbf Q_{\it WE, M}^\ast$ such that $P = QT^{-1}$. We then have $Q\{Z = z\} = P\{Z = z\}$ for all $z \in \mathcal Z$, and because $Q$ satisfies Assumption \ref{as:marginal-exog}, we know \eqref{eq:p-q} holds. Let $Q_1$ denote the marginal distribution of $((Y_d: d \in \mathcal D), (D_z: z \in \mathcal Z))$ under $Q$, and $Q_Z$ denote the marginal distribution of $Z$ under $Q$, and define $\tilde Q = Q_1 \times Q_Z$. Then, $\tilde Q$ satisfies Assumption \ref{as:exog} by construction, and it satisfies Assumption \ref{as:encouragement} because the marginal distribution of $(D_z: z \in \mathcal Z)$ under $\tilde Q$ is the same as that under $Q$. In summary, $\tilde Q \in \mathbf Q_{E, M}^\ast$. Furthermore, $P = \tilde Q T^{-1}$ because (1) $\tilde Q\{Z = z\} = Q\{Z = z\} = P \{Z = z\}$ for all $z \in \mathcal Z$; and (2) $\tilde Q\{Y_d = y, D_z = d\} = Q\{Y_d = y, D_z = d\}$ for all $d \in \mathcal D$ and $z \in \mathcal Z$, and hence \eqref{eq:p-q} is still satisfied. As a result, we know $\mathbf Q_0(P, \mathbf Q_{E, M}^\ast) \neq \emptyset$.

Next, take $\mathbf Q = \mathbf Q_{E, M}^\ast$ and $\mathbf Q' = \mathbf Q_{\it WE, M}^\ast$, and note that $\mathbf Q_{E, M}^\ast \subseteq \mathbf Q_{\it WE, M}^\ast \subseteq \mathbf Q_{\it MI}^\ast$ because Assumption \ref{as:exog} implies Assumption \ref{as:marginal-exog}, which in turn implies Assumption \ref{as:mean_iv}. Because we know $\mathbf Q_0(P, \mathbf Q_{E, M}^\ast) \neq \emptyset$ from the previous paragraph, we then obtain from Theorem \ref{theorem:same} that
\[ \Theta_0(P, \mathbf Q_{E, M}^\ast) = \Theta_0(P, \mathbf Q_{\it WE,M}^\ast) = \Theta_0(P, \mathbf Q_{\it MI}^\ast)~. \]
Similarly, taking $\mathbf Q' = \mathbf Q_{\it WE}^\ast$, we have
\[ \Theta_0(P, \mathbf Q_{E, M}^\ast) = \Theta_0(P, \mathbf Q_{\it WE}^\ast) = \Theta_0(P, \mathbf Q_{\it MI}^\ast)~. \]
The desired conclusion now follows. \qed

\section{Details of Examples} 

\subsection{Details of Example \ref{ex:lee2023}}\label{sec:lee2023}

We show that,  except in the special case where the treatment and the instrument are both binary,    the strict one-to-one targeting assumption of  \cite{lee2023treatment} with one or more targeted treatments implies that \eqref{eq:uniform} does not hold for some treatments. To see this, suppose that the strict one-to-one targeting assumption of  \cite{lee2023treatment} holds with  $| \mathcal{D}^\dagger| \ge 1$. For each $d \in \mathcal{D}^\dagger$, their assumptions include   that there exists some $z^\dagger(d)$ and some $ \overline{U}(d) $, $\underline{U}(d)$ with $ \overline{U}(d) > \underline{U}(d)$ such that
\begin{align}  \label{eq:stricttarget}
g(z,d) = & 
\begin{cases} \overline{U}(d) & \mbox{if}~ z=z^\dagger(d) \\
\underline{U}(d) & \mbox{if}~ z \ne z^\dagger(d)~.
\end{cases}
\end{align}
On the other hand, for each $d \in \mathcal{D} \setminus \mathcal{D}^\dagger$, they impose that $g(z,d) = \underline{U}(d)$ for all $z \in \mathcal{Z}$, and they impose that there is at least one such non-targeted treatment. For any  $d \in \mathcal{D} \setminus \mathcal{D}^\dagger$,  \eqref{eq:uniform}  requires that there exists $z^\ast(d)\in\mathcal{Z}$ such that
\begin{equation}  \label{eq:lee2023ineq}
g(z, d') >   g(z^\ast(d), d')  ~~  \text{ for all } d' \neq d \text{ and } z \neq z^\ast(d)~.
\end{equation}

Suppose $|\mathcal{Z}| \ge 3$, and fix some targeted treatment $d' \in \mathcal{D}^\dagger$.  Suppose $z^\ast(d) \ne z^\dagger(d')$.  Then, for $z \in \mathcal{Z} \setminus \{z^\ast(d),z^\dagger(d')\}$, \eqref{eq:lee2023ineq} requires $\underline{U}(d')> \underline{U}(d')$, a contradiction. Now suppose $z^\ast(d) = z^\dagger(d')$.  Then, for $z \in \mathcal{Z} \setminus \{z^\ast(d)\}$, \eqref{eq:lee2023ineq} requires   $\underline{U}(d')> \overline{U}(d')$, a contradiction. Thus, $|\mathcal{Z}| \ge 3$ implies that \eqref{eq:uniform} does not hold for non-targeted treatments. 

Now suppose  $|\mathcal{Z}| = 2$, which we label as $\mathcal Z = \{0, 1\}$, and suppose $|\mathcal{D}| \ge 3$.   Without loss of generality suppose  $1 \in  \mathcal D^\dagger$ and $z^\dagger(1) = 1$.  If $z^\ast(d)=1$, then \eqref{eq:lee2023ineq} requires  $\underline U(1) > \overline U(1)$, a contradiction.  Now suppose $z^\ast(d)=0$. Then \eqref{eq:lee2023ineq} requires  $g(1, d') > g(0, d')$ for all $d' \neq d$. Consider the following two cases:
\begin{itemize}
    \item If $|\mathcal D^\dagger| = 1$, then  $g(1, d') > g(0, d')$ holding for any $d' \in (\mathcal D \setminus \mathcal D^\dagger) \setminus \{d\}$ requires $\underline U(d') > \underline U(d')$, a contradiction.
    \item If $|\mathcal D^\dagger| > 1$, then there exists $d'' \in D^\dagger \setminus \{1\}$. By assumption $z^\dagger(d'') \neq z^\dagger(1)$ so that $z^\dagger(d'') = 0$. Then $g(1, d') > g(0, d')$ holding for  $d' = d''$ requires $\underline U(d'') > \overline U(d'') $, again a contradiction.
\end{itemize}
Thus, $|\mathcal{Z}| = 2$ with $|\mathcal{D}| \ge 3$  implies that \eqref{eq:uniform} does not hold for some treatments. 

We have shown \eqref{eq:uniform} does not hold for some treatments when either $Z$ or $D$ takes at least three values. Now suppose   $|\mathcal{D}| = |\mathcal{Z}|=2$.  Let $D=0$ denote the nontargeted treatment and $D=1$ the targeted treatment, and let $z^{\dagger}(1)=1$. Consider  $z^\ast (0) = 0$ and $z^\ast (1) = 1$.  Then evaluating  \eqref{eq:uniform} at either $d=0$ or $d=1$ results in $ \overline U(1) > \underline U(1),$ and thus \eqref{eq:uniform} holds when $|\mathcal{D}| = |\mathcal{Z}|=2$. We conclude that the strict one-to-one targeting of  \cite{lee2023treatment} implies that \eqref{eq:uniform} does not hold for some $d \in \mathcal{D}$ except in the special case where $|\mathcal{D}| = |\mathcal{Z}|=2$.

We now show that the strict one-to-one targeting of  \cite{lee2023treatment} implies that Assumption \ref{as:encouragement} holds when $|\mathcal{Z}| > |\mathcal{D}^\dagger|$.   Let $\mathcal{Z}^\dagger \subseteq \mathcal{Z}$ denote the set of instruments that target some treatment, $$\mathcal{Z}^\dagger = \{ z \in \mathcal{Z} : z = z^\dagger(d) ~ \mbox{for some}~d \in \mathcal{D}^\dagger\}~.$$ Their strict one-to-one targeting assumption combined with $|\mathcal{Z}| > |\mathcal{D}^\dagger|$ implies that there are values of the instrument that do not target any treatment; in other words, $\mathcal{Z}^\dagger \subsetneqq  \mathcal{Z}$. Following \cite{lee2023treatment},  we label the treatment that is known not to be targeted   as treatment $0$, so that  $g(z,0) = \underline U(0)$ for all $z \in \mathcal Z$, and impose their normalization that $\underline U(0)=0$. Consider \eqref{eq:encouragement} for $d=0$. Note that 
$$Q\{D_{z^\ast(0)} \ne 0 ,   ~D_{z^{\prime}}=0 ~ \mbox{for some}  ~ z^{\prime} \ne z^\ast(0) \}
=~Q \left \{ \bigcup_{d^\ast  \ne  0, z^{\prime} \ne z^\ast(0)}  D_{z^\ast(0)} = d^\ast ,   ~D_{z^{\prime}}=0 \right \} ~.
$$ 
We wish to investigate whether there exists some  $z^\ast(0) \in \mathcal{Z}$ such that the above probability is zero.  Consider $z^\ast(0)$ equal to any value in $\mathcal{Z} \setminus \mathcal{Z}^\dagger$, i.e., any value of the instrument that does not target any treatment.
For any fixed $d^\ast  \ne  0, z^{\prime} \ne z^\ast(0)$, consider the event $\{ D_{z^\ast(0)} = d^\ast ,   ~D_{z^{\prime}}=0\}$. Since $z^\ast(0)$ does not target any treatment and thus does not target $d^\ast$, $D_{z^\ast(0)}=d^\ast$ implies
\begin{equation}\label{eq:a'}
U_0- U_{d^\ast} \le  \underline{U}(d^\ast) ~.
\end{equation}
If $z^{\prime}$ targets $d^\ast$, then $D_{z^{\prime}}=0$
implies
\begin{equation}\label{eq:b}
U_0 - U_{d^\ast}\ge  \overline{U}(d^\ast)  ~.
\end{equation} 
If $z^{\prime}$ does not target $d^\ast$, then $D_{z^{\prime}}=0$ implies
\begin{equation}\label{eq:c}
U_0 - U_{d^\ast}\ge  \underline{U}(d^\ast)~.
\end{equation} 
Thus, the event $\{ D_{z^\ast(0)} = d^\ast ,   ~D_{z^{\prime}}=0 \}$ either requires \eqref{eq:a'} and \eqref{eq:b} to jointly hold, which is a contradiction since $\overline{U}(d^\ast)> \underline{U}(d^\ast)$ , or requires \eqref{eq:a'} and \eqref{eq:c} to jointly hold, which is a zero probability event given our assumption that the distribution of $(U_d: d \in \mathcal D)$ is absolutely continuous w.r.t.\ Lebesgue measure. Thus $Q\{D_{z^\ast(0)} \ne 0 ,   ~D_{z^{\prime}}=0 ~ \mbox{for some}  ~ z^{\prime} \ne z^\ast(0) \}$ is a probability of a finite union of zero probability events, and thus, by Boole's inequality, equals zero so that \eqref{eq:encouragement} holds for $d=0$. A parallel argument shows that \eqref{eq:encouragement} holds for any non-targeted treatment, and related argument shows that  \eqref{eq:encouragement} holds for any targeted treatment.  Thus, under the strict one-to-one targeting of  \cite{lee2023treatment}, when there are more values of the instrument than targeted treatments, Assumption \ref{as:encouragement} holds even though \eqref{eq:uniform} is violated for some $d \in \mathcal{D}$. 

\subsection{Details of Example \ref{eg:alt}} \label{sec:alt}
Let $\mathbf{Q}$ denote all distributions $Q$ for which Assumption \ref{as:exog} holds and such that $Q\{ D_0 = D_1\} =0$.  Then for   $Q \in \mathbf{Q}_0(P, \mathbf{Q})$,
\begin{align*}
p_{y1|1} &= Q\{Y_1=y, D_1 =1, D_0 =0\}\\
p_{y0|0} &= Q\{Y_0=y, D_1 =1, D_0 =0\}\\
p_{y0|1} &= Q\{Y_0=y, D_1 =0, D_0 =1\}\\
p_{y1|0} &= Q\{Y_1=y, D_1 =0, D_0 =1\}
\end{align*}  
and
\begin{align*}
Q\{Y_0=1\} &= Q\{Y_0=1, D_1 =1, D_0 =0\} + Q\{Y_0=1, D_1 =0, D_0 =1\}\\
    &= p_{10|0}  + p_{10|1}\\
    Q\{Y_1=1\} &= Q\{Y_1=1, D_1 =1, D_0 =0\} + Q\{Y_1=1, D_1 =0, D_0 =1\}\\
    &= p_{11|1}  + p_{11|0}~.
\end{align*}  
Therefore, if $\mathbf{Q}$ is consistent with $P$, then $\theta(Q)$ is identified as
\begin{equation} \label{eq:point}
    \Theta_0(P, \mathbf Q) = \left \{ \begin{pmatrix}
    p_{10|0} + p_{10|1} \\ p_{11|0} + p_{11|1}
\end{pmatrix} \right \}~.
\end{equation}

In contrast, the identified set that follows from imposing Assumption \ref{as:exog} alone, $\Theta_0(P, \mathbf Q_E^\ast)$, is shown by \cite{balke1997bounds} to be
\begin{equation}\label{eq:bp_bounds_y0}
    \max \begin{Bmatrix} p_{10|1} \\ p_{10|0} \\ p_{10|0} + p_{11|0} - p_{00|1} - p_{11|1} \\ p_{01|0} + p_{10|0} - p_{00|1} - p_{01|1} \end{Bmatrix} \leq Q\{Y_0=1\} \leq \min \begin{Bmatrix}
        1-p_{00|1}\\ 1-p_{00|0} \\ p_{01|0} + p_{10|0} + p_{10|1} + p_{11|1} \\ p_{10|0} + p_{11|0} + p_{01|1} + p_{10|1}
    \end{Bmatrix}
\end{equation}
and 
\begin{equation}\label{eq:bp_bounds_y1}
    \max \begin{Bmatrix} p_{11|0} \\ p_{11|1} \\ -p_{00|0} - p_{01|0} + p_{00|1} + p_{11|1} \\ -p_{01|0} - p_{10|0} + p_{10|1} + p_{11|1} \end{Bmatrix} \leq Q\{Y_1=1\} \leq \min \begin{Bmatrix}
        1-p_{01|1}\\ 1-p_{01|0} \\ p_{00|0} + p_{11|0} + p_{10|1} + p_{11|1} \\ p_{10|0} + p_{11|0} + p_{00|1} + p_{11|1}
    \end{Bmatrix}~. 
\end{equation}
It follows from  $ \mathbf Q \subseteq \mathbf Q_E^\ast$ that $\Theta_0(P, \mathbf Q) \subseteq \Theta_0(P, \mathbf Q_E^\ast)$, and thus \eqref{eq:point} is contained in \eqref{eq:bp_bounds_y0}--\eqref{eq:bp_bounds_y1}.

Next, we show that there exists a $P$ for which $\mathbf Q_0(P, \mathbf Q) \neq \emptyset$, $\Theta_0(P, \mathbf Q)$ is not given by \eqref{eq:bounds1} and $\Theta_0(P, \mathbf Q) \subsetneqq \Theta_0(P, \mathbf Q_E^\ast)$.  We do so by providing a numerical example. Consider the $P$ specified in Table \ref{tab:distP_Ex51} and the $Q$ specified in Table \ref{tab:distQ_Ex51}, where we write $q(y_0 y_1, d_0 d_1) = Q\{Y_d = y_d, D_z = d_z, ~(d,z) \in \mathcal D \times \mathcal Z\}$ and omit any $q(\cdot)=0$. One can check that $Q \in \mathbf{Q}$ and rationalizes $P$, so that $Q \in  \mathbf Q_0(P, \mathbf Q) \neq \emptyset$. In this example,   $Q\{Y_0=1\} = 0.4274$, and thus the identified set for $Q\{Y_0=1\}$ relative to $ \mathbf Q$ is the singleton $\{  0.4274\}$.  In contrast, evaluating  \eqref{eq:bp_bounds_y0} at $P$ gives the  identified set for $Q\{Y_0=1\}$ relative to $\mathbf Q_E^\ast$ as $[0.3336, 0.5212]$.  We thus conclude that $\Theta_0(P, \mathbf Q) \subsetneqq \Theta_0(P, \mathbf Q_E^\ast)$ for some $P$ that can be rationalized by $Q \in \mathbf{Q}$.  Now consider evaluating the bounds of \eqref{eq:bounds1} at the same $P$.  Doing so results in bounds on $Q\{Y_0=1\}$ given by $[0.1618, 0.5445]$ if setting $z^\ast(0) = 0$ and given by $[0.2656, 0.8829]$ if setting $z^\ast(0) = 1$.  Therefore, no matter $z^\ast(0) = 0$ or $1$, the bounds \eqref{eq:bounds1}   is not the identified set for $Q\{Y_0=1\}$ relative to either $\mathbf Q_E^\ast$ or $\mathbf Q$.   

\begin{table}[ht!]
    \centering
    \small
    \begin{tabular}{|c|c|c|c|}
    \hline
         $p_{ 00|0 }$ & $p_{ 10|0 }$ & $p_{ 01|0 }$ & $p_{ 11|0 }$ \\
         0.4555 & 0.1618 & 0.3077 & 0.0750 \\
         \hline
         $p_{ 00|1 }$ & $p_{ 10|1 }$ & $p_{ 01|1 }$ & $p_{ 11|1 }$ \\
         0.1171 & 0.2656 & 0.0188 & 0.5985 \\
    \hline
    \end{tabular}
    \caption{Distribution $P$ in Appendix \ref{sec:alt}.}
    \label{tab:distP_Ex51}
\end{table}

\begin{table}[ht!]
    \centering
    \small
    \begin{tabular}{|c|c|c|c|}
    \hline
         $q( 00,01 )$ & $q( 00,10 )$ & $q( 01,01 )$ & $q( 01,10 )$ \\
         0.0039 & 0.0428 & 0.4516 & 0.0743 \\
         \hline
         $q( 10,01 )$ & $q( 10,10 )$ & $q( 11,01 )$ & $q( 11,10 )$ \\
         0.0149 & 0.2649 & 0.1469 & 0.0007 \\
    \hline
    \end{tabular}
    \caption{Distribution $Q$.}

    \label{tab:distQ_Ex51}
\end{table}

\subsection{Details of Example \ref{ex:ordered}}\label{sec:exordered}
Suppose $\mathcal Y = \{0, 1\}$, $\mathcal D = \{0, 1, 2\}$, and $\mathcal Z = \{0, 1, 2\}$. Then, the linear program approach in \cite{balke1993nonparametric,balke1997bounds} leads to the following identified set for $\mathbb E_Q[Y_1] = Q\{Y_1 = 1\}$ relative to $\mathbf Q$ being defined as in Example \ref{ex:ordered}:
\begin{equation} \label{eq:1-ordered}
    \left[ \max\begin{Bmatrix}
        p_{11|0}\\ p_{11|1}\\p_{11|2}\\ p_{11|0}-p_{11|1}+p_{11|2}
    \end{Bmatrix}, \quad \min\begin{Bmatrix}
           1-p_{01|2}\\1-p_{01|1}\\1-p_{01|0}\\1-p_{01|0}+p_{01|1}-p_{01|2}
    \end{Bmatrix} \right]~.
\end{equation}
We will show that for some $P$ such that $\mathbf Q_0(P, \mathbf Q) \neq \emptyset$, we have $\Theta_0(P, \mathbf Q)$ strictly smaller than \eqref{eq:bounds1} and $\Theta_0(P, \mathbf Q^\ast_1)$. For this purpose we are only concerned with the validity of \eqref{eq:1-ordered} instead of its sharpness. For the lower bounds, first note for $z \in \mathcal Z$,
\[ Q\{Y_1 = 1\} = Q\{Y_1 = 1 \mid Z = z\} \geq Q\{Y_1 = 1, D = 1 \mid Z = z\} = Q \{Y = 1, D = 1 \mid Z = z\}~, \]
and therefore the first three rows follow. To show the last row, note it's equivalent to
\[ Q\{Y_1 = 1, D_1 = 1\} + Q\{Y_1 = 1, D_0 = 0\} + Q\{Y_1 = 1, D_0 = 2\} \geq Q\{Y_1 = 1, D_2 = 1\}~. \]
It therefore suffices to show that
\begin{equation} \label{eq:ordered-implies}
    \{D_2 = 1\} \implies \{D_1 = 1\} \cup \{D_0 = 0\} \cup \{D_0 = 2\}~.
\end{equation}
Suppose $D_2 = 1$ but $D_0 \neq 0$ and $D_0 \neq 2$. Then $D_0 = 1$. But $D_0 \leq D_1 \leq D_2$, so $D_1 = 1$. \eqref{eq:ordered-implies} now follows. The lower bounds in \eqref{eq:1-ordered} have all been shown to hold, and the upper bounds can be proved similarly.

Next, we show that there exists a $P$ for which $\mathbf Q_0(P, \mathbf Q) \neq \emptyset$, $\Theta_0(P, \mathbf Q)$ is not given by \eqref{eq:bounds1} and $\Theta_0(P, \mathbf Q) \subsetneqq \Theta_0(P, \mathbf Q_E^\ast)$. We do so by providing a numerical example. Consider the $P$ specified in Table~\ref{tab:distP_Ex52} and the four $Q$ distributions specified in Table~\ref{tab:distQ_Ex52_exogmin}, \ref{tab:distQ_Ex52_exogmax}, \ref{tab:distQ_Ex52_exoord_min} and \ref{tab:distQ_Ex52_exoord_max}, which we denote as $Q_{\text{ex},\text{min}}$, $Q_{\text{ex},\text{max}}$, $Q_{\text{ex},\text{om},\text{min}}$ and $Q_{\text{ex},\text{om},\text{max}}$ respectively, where we write $q(y_0 y_1 y_2, d_0 d_1 d_2) = Q\{Y_d = y_d, D_z = d_z, ~(d,z) \in \mathcal D \times \mathcal Z\}$ and omit any $q(\cdot)=0$. One can check that all the four $Q$s are in $\mathbf Q_0(P, \mathbf Q^\ast_E)$, i.e., they all rationalize $P$ and satisfy Assumption \ref{as:exog}. Moreover, $Q_{\text{ex},\text{om},\text{min}} \in \mathbf Q_0(P, \mathbf Q)$ and $Q_{\text{ex},\text{om},\text{max}} \in \mathbf Q_0(P, \mathbf Q)$ so that $\mathbf Q_0(P, \mathbf Q) \neq \emptyset$. Evaluating \eqref{eq:1-ordered} at $P$ gives $[0.2117, 0.8205] =: I_{\text{ex}, \text{om}}$. In contrast, if one evaluates \eqref{eq:bounds1} by setting $z^\ast(1) = 0, 1, 2$ at the same $P$, the resulting bounds for $\mathbb E_Q[Y_1]$ are $[0.1664, 0.9255] =: I_{\eqref{eq:bounds1}, 0}$, $[0.0712, 0.9311] =: I_{\eqref{eq:bounds1}, 1}$ and $[0.1165, 0.8261] =: I_{\eqref{eq:bounds1}, 2}$ respectively. In all cases, we see $I_{\text{ex}, \text{om}} \subsetneqq I_{\eqref{eq:bounds1}, z^\ast(1)}$ for all $z^\ast(1) \in \{0,1,2\}$ so $\Theta_0(P, \mathbf Q)$ is not given by \eqref{eq:bounds1}. Furthermore, $Q_{\text{ex},\text{min}} \notin \mathbf Q_0(P, \mathbf Q)$ and $Q_{\text{ex},\text{max}} \notin \mathbf Q_0(P, \mathbf Q)$ because, for example, $q_{\text{ex},\text{min}}(000,021) > 0$ and $q_{\text{ex},\text{max}}(000,210) > 0$. At the same time, $\mathbb E_{Q_{\text{ex},\text{min}}}[Y_1] = 0.1664 \notin I_{\text{ex}, \text{om}}$ and $\mathbb E_{Q_{\text{ex},\text{max}}}[Y_1] = 0.8261 \notin I_{\text{ex}, \text{om}}$. Therefore, $\Theta_0(P, \mathbf Q) \subsetneqq \Theta_0(P, \mathbf Q_E^\ast)$. 

\begin{table}[ht!]
    \centering
    \small
    \begin{tabular}{|c|c|c|c|c|c|}
    \hline
        $p_{ 00|0 }$ & $p_{ 10|0 }$ & $p_{ 01|0 }$ & $p_{ 11|0 }$ & $p_{ 02|0 }$ & $p_{ 12|0 }$ \\
        0.3808 & 0.2427 & 0.0745 & 0.1664 & 0.0345 & 0.1011 \\
        \hline
        $p_{ 00|1 }$ & $p_{ 10|1 }$ & $p_{ 01|1 }$ & $p_{ 11|1 }$ & $p_{ 02|1 }$ & $p_{ 12|1 }$ \\
        0.2830 & 0.1947 & 0.0689 & 0.0712 & 0.2014 & 0.1808 \\
        \hline
        $p_{ 00|2 }$ & $p_{ 10|2 }$ & $p_{ 01|2 }$ & $p_{ 11|2 }$ & $p_{ 02|2 }$ & $p_{ 12|2 }$\\
        0.0802 & 0.0976 & 0.1739 & 0.1165 & 0.2444 & 0.2874 \\ 
    \hline
    \end{tabular}
    \caption{Distribution $P$ in Appendix \ref{sec:exordered}.}
    \label{tab:distP_Ex52}
\end{table}

\begin{table}[ht!]
    \centering
    \small
    \begin{tabular}{|c|c|c|c|c|c|c|c|}
    \hline
        $q( 000,002 )$ & $q( 000,011 )$ & $q( 000,021 )$ & $q( 001,002 )$ & $q( 001,020 )$ & $q( 001,021 )$ & $q( 001,121 )$ & $q( 001,211 )$ \\
        0.0139 & 0.0304 & 0.0054 & 0.2238 & 0.0802 & 0.0271 & 0.0735 & 0.0375 \\
        \hline
        $q( 010,101 )$ & $q( 010,111 )$ & $q( 010,122 )$ & $q( 100,000 )$ & $q( 100,022 )$ & $q( 100,110 )$ & $q( 100,202 )$ & $q( 101,202 )$ \\
        0.0453 & 0.0712 & 0.0499 & 0.0966 & 0.1461 & 0.0010 & 0.0345 & 0.0636 \\ 
    \hline
    \end{tabular}
    \caption{Distribution $Q_{\text{ex},\text{min}}$.}
    \label{tab:distQ_Ex52_exogmin}
\end{table}

\begin{table}[ht!]
    \centering
    \small
    \begin{tabular}{|c|c|c|c|c|c|c|c|}
    \hline
        $q( 001,021 )$ & $q( 001,121 )$ & $q( 001,211 )$ & $q( 010,000 )$ & $q( 010,001 )$ & $q( 010,002 )$ & $q( 010,010 )$ & $q( 010,122 )$ \\
        0.0305 & 0.0745 & 0.0689 & 0.0220 & 0.0064 & 0.0085 & 0.0260 & 0.1664 \\
        \hline
        $q( 010,202 )$ & $q( 011,002 )$ & $q( 011,022 )$ & $q( 011,210 )$ & $q( 110,000 )$ & $q( 110,001 )$ & $q( 110,011 )$ & $q( 110,022 )$ \\
        0.0345 & 0.2116 & 0.0758 & 0.0322 & 0.0976 & 0.0971 & 0.0130 & 0.0350 \\
    \hline
    \end{tabular}
    \caption{Distribution $Q_{\text{ex},\text{max}}$.}
    \label{tab:distQ_Ex52_exogmax}
\end{table}

\begin{table}[ht!]
    \centering
    \small
    \begin{tabular}{|c|c|c|c|c|c|c|c|}
    \hline
        $q( 000,000 )$ & $q( 000,001 )$ & $q( 000,002 )$ & $q( 000,022 )$ & $q( 000,111 )$ & $q( 000,122 )$ & $q( 000,222 )$ & $q( 001,002 )$ \\
        0.0802 & 0.0079 & 0.0430 & 0.0181 & 0.0209 & 0.0536 & 0.0345 & 0.1066 \\
        \hline
        $q( 001,022 )$ & $q( 001,222 )$ & $q( 010,001 )$ & $q( 010,111 )$ & $q( 010,122 )$ & $q( 100,000 )$ & $q( 100,001 )$ & $q( 100,011 )$ \\
        0.0797 & 0.1011 & 0.0453 & 0.0712 & 0.0952 & 0.0976 & 0.0971 & 0.0480 \\ 
    \hline
    \end{tabular}
    \caption{Distribution $Q_{\text{ex},\text{om},\text{min}}$.}
    \label{tab:distQ_Ex52_exoord_min}
\end{table}

\begin{table}[ht!]
    \centering
    \small
    \begin{tabular}{|c|c|c|c|c|c|c|c|}
    \hline
        $q( 000,001 )$ & $q( 000,111 )$ & $q( 000,122 )$ & $q( 010,000 )$ & $q( 010,001 )$ & $q( 010,012 )$ & $q( 010,111 )$ & $q( 010,122 )$ \\
        0.0079 & 0.0689 & 0.0056 & 0.0802 & 0.1114 & 0.0430 & 0.0051 & 0.1613 \\
        \hline
        $q( 010,222 )$ & $q( 011,002 )$ & $q( 011,022 )$ & $q( 011,222 )$ & $q( 100,001 )$ & $q( 110,000 )$ & $q( 111,012 )$ & $q( 111,022 )$ \\
        0.0345 & 0.0835 & 0.0548 & 0.1011 & 0.0971 & 0.0976 & 0.0231 & 0.0249 \\
    \hline
    \end{tabular}
    \caption{Distribution $Q_{\text{ex},\text{om},\text{max}}$.}
    \label{tab:distQ_Ex52_exoord_max}
\end{table}

\subsection{Details of Example \ref{ex:arum-violation}}\label{sec:arum-violation}
Consider the ARUM defined by \eqref{eq:ARUM} with strict  one-to-one targeting  and  $|\mathcal{D} | = 3$, $|\mathcal{Z}| = |\mathcal{D}^\dagger|=2$.  In this case, there are two targeted treatments and one non-targeted treatment.  Following \cite{lee2023treatment}, label that non-targeted treatment as treatment $0$ and impose the normalization that $g(z, 0) = 0$ for all $z \in \mathcal Z$. Label $Z=0$ as the instrument value that targets treatment $1$ and label $Z=1$ as the instrument value that targets treatment $2$, so that \eqref{eq:stricttarget}  holds for $d=1,2$ for some $ \overline{U}(d), \underline{U}(d)$ with  $ \overline{U}(d) > \underline{U}(d)$  and with $z^\dagger(1)=0$,  $z^\dagger(2)=1$. Let $\mathbf Q$ denote the set of all distributions for which $(D_z: z \in \mathcal Z)$ is determined by \eqref{eq:ARUM} with these restrictions and additionally imposing that the support of $(U_0,U_1,U_2)$ is $ \Re^3$. Recall $(U_0, U_1, U_2) \indep Z$ by assumption.

Let $U_{10}=U_1-U_0$ and $U_{20}=U_2-U_0$. In this model, the treatment value is completely determined by the vector of realizations $(U_{10}, U_{20})$.  For instance,  $D_z = 2$ if and only if
\begin{align*}
    U_{20} & \geq - g(z, 2) \\
    U_{20} - U_{10} & \geq g(z, 1) - g(z, 2)~,
\end{align*}
and a similar characterization holds for $D_z = 1$.   See Figure \ref{fig:012}, which is taken from Figure 1 in \cite{lee2023treatment}. 
\begin{figure}[ht!]
    \centering
    \resizebox{0.5\columnwidth}{!}{
\begin{tikzpicture}
    \fill[blue!20] (0, 0) -- (0, -4) -- (4, -4) -- (4, 4);
    \fill[pink] (0, 0) -- (-4, 0) -- (-4, -4) -- (0, -4);
    \fill[green!20] (0, 0) -- (-4, 0) -- (-4, 4) -- (4, 4);
    \draw[] (0, 0) node[below right]{$(-g(z, 1), -g(z, 2))$} to (4, 4);
    \draw[->] (-4, 0) to (4, 0) node[below]{$u_{10}$};
    \draw[->] (0, -4) to (0, 4) node[right]{$u_{20}$};
    \draw[] node at (2, -2) {$D_z = 1$};
    \draw[] node at (-2, 2) {$D_z = 2$};
    \draw[] node at (-2, -2) {$D_z = 0$};
\end{tikzpicture}
}
    \caption{Treatment under each value of $(u_{10}, u_{20})$ for a given $z$.}
    \label{fig:012}
\end{figure}

We first show that \eqref{eq:encouragement} holds  for the targeted treatments by verifying that \eqref{eq:uniform}  holds for the targeted treatments.  Consider  \eqref{eq:uniform} for $d =1$. It holds with $z^\ast(1)=z^{\dagger}(1)=0$ because
\begin{align*}
\overline{U}(1) > &~ \underline{U}(1)~,\\
\overline{U}(1) - \underline{U}(2) > &~ \underline{U}(1) -\overline{U}(2)~, 
\end{align*}
which in turn holds because $ \overline{U}(d) > \underline{U}(d)$ for $d \in \{1,2\}$. Thus \eqref{eq:uniform} holds for $d=1$, which, as shown in Example \ref{ex:arum}, implies that \eqref{eq:encouragement} holds for $d=1$. By a parallel argument, \eqref{eq:encouragement} holds for $d=2$. 

We now show that there does not exist a value of $z^{\ast}(0)$ such that \eqref{eq:encouragement} holds for the non-targeted treatment, treatment $0$.   Suppose $z^{\ast}(0)=0$. Then 
\begin{align*}  Q\{D_{0} \ne 0 ,   ~D_{1}=0 \} 
\ge & ~  Q\{D_{0} =1,   ~D_{1}=0 \} \\
 = & ~Q\{-\underline{U}(1) \ge U_{10} \ge -\overline{U}(1) , U_{20} \le -\overline{U}(2), U_{10}-U_{20} \ge \underline{U}(2)-\overline{U}(1) \}\\
 > &~0~,
\end{align*}
where the last line is using that the support of the distribution of $(U_{10},U_{20})= \Re^2$ by assumption and that   strict targeting of treatment $1$ requires  $-\underline{U}(1)> -\overline{U}(1)$.  Thus \eqref{eq:encouragement} cannot hold for $d=0$ with $z^{\ast}(0)=0$. A parallel argument shows that \eqref{eq:encouragement} cannot hold for $d=0$ with $z^{\ast}(0)=1$. 

We conclude that,  when $|\mathcal{D}|=3$ and $|\mathcal{D}^\dagger|=|\mathcal{Z}|=2$, one-to-one strict targeting with the regularity condition that the support of $(U_0,U_1,U_2)$ is $\Re^3$ implies that \eqref{eq:encouragement} holds for the targeted treatments but not for the non-targeted treatments, and thus  Assumption \ref{as:encouragement}  cannot hold.  This argument  can be adapted for any ARUM with $|\mathcal{D}| \ge 3$ and $|\mathcal{Z}|= |\mathcal{D}^\dagger|$ to show that, while \eqref{eq:encouragement} holds for the targeted treatments, there does not exist a value of $z^{\ast}(d)$ such that \eqref{eq:encouragement} holds 
for any non-targeted treatment $d$, and thus that Assumption \ref{as:encouragement}  cannot hold.

\begin{figure}[ht!]
    \centering
    \resizebox{0.5\columnwidth}{!}{
\begin{tikzpicture}
    \fill[blue!20] (1, -1) -- (1, -4) -- (4, -4) -- (4, 2);
    \fill[pink] (0, -1) -- (-4, -1) -- (-4, -4) -- (0, -4);
    \fill[green!20] (0, 0) -- (-4, 0) -- (-4, 4) -- (4, 4);
    \fill[cyan!20] (0, 0) -- (0, -1) -- (1, -1) -- (4, 2) -- (4, 4);
    \fill[magenta!20] (0, 0) -- (-4, 0) -- (-4, -1) -- (0, -1);
    \fill[yellow!20] (0, -1) -- (0, -4) -- (1, -4) -- (1, -1);
    \draw[] (0, 0) node[below left]{$z = 0$} to (4, 4);
    \draw[] (1, -1) node[below right]{$z = 1$} to (4, 2);
    \draw[] (0, 0) to (-4, 0);
    \draw[] (0, 0) to (0, -4);
    \draw[] (1, -1) to (-4, -1);
    \draw[] (1, -1) to (1, -4);
    \draw[] node at (2.5, -2) {$(1, 1)$};
    \draw[] node at (-2, 2) {$(2, 2)$};
    \draw[] node at (-2, -2.5) {$(0, 0)$};
    \draw[] node at (2, 1) {$(1, 2)$};
    \draw[] node at (-2, -0.5) {$(0, 2)$};
    \draw[] node at (0.5, -2.5) {$(1, 0)$};
\end{tikzpicture}
}
    \caption{Values of $(D_0, D_1)$ for each value of $(u_{10}, u_{20})$.}
    \label{fig:012-shift}
\end{figure}

Next, consider the identified sets for the average potential outcomes.
Since \eqref{eq:encouragement} is satisfied for  the targeted treatments,  a straightforward modification of the arguments underlying Theorem \ref{theorem:ev} show that the identified sets for $\mathbb{E}_Q[Y_d]$  for $d \in \{1,2\}$ is given  by \eqref{eq:bounds1} for any $P$ such that $\mathbf{Q}_0(P, \mathbf{Q}) \neq \emptyset$.

We now derive the identified set for the average potential outcome of the non-targeted treatment. First, note that $-g(0, 1) = -\overline U(1) < - \underline U(1) = -g(1, 1)$ and $-g(0, 2) = -\underline U(2) > -\overline U(2) = -g(1, 2)$. Therefore, it can be verified from Figure \ref{fig:012-shift} that for all $Q \in \mathbf Q$,
\begin{equation} \label{eq:arum-violation-rt}
    Q\{(D_0, D_1) \in \{(0, 0), (1, 0), (1, 1), (0, 2), (1, 2), (2, 2)\}\} = 1~.
\end{equation}
Let $\mathbf Q'$ denote the set of all distributions that satisfies \eqref{eq:arum-violation-rt}. Note that all $Q \in \mathbf Q$ satisfies \eqref{eq:arum-violation-rt}, so $\mathbf Q \subseteq \mathbf Q'$. On the other hand,  by assigning appropriate probabilities to each set in the partition in Figure \ref{fig:012-shift}, we immediately see that each $Q \in \mathbf Q'$ can be rationalized by a $Q \in \mathbf Q$. Therefore, $\mathbf Q = \mathbf Q'$. Using linear programming as in \cite{balke1993nonparametric,balke1997bounds}, we obtain the following identified set for $\mathbb E_Q[Y_0] = Q \{Y_0 = 1\}$ relative to $\mathbf Q$:
\begin{equation}\label{eq:ARUM-3arm-Y0}
    \left [ \max\begin{Bmatrix}
        p_{10|0} \\ p_{10|1}
    \end{Bmatrix}, \quad \min\begin{Bmatrix}
           1-p_{00|1}\\1-p_{00|0}
    \end{Bmatrix} \right ]~.
\end{equation}
The identified set in \eqref{eq:ARUM-3arm-Y0} equals \eqref{eq:meaniv2} for $d=0$ with $Y$ and $Z$ binary.  Thus, the identified set  for $\mathbb E_Q[Y_0]$ relative to $\mathbf{Q}$ corresponds to the identified set relative to $\mathbf{Q}^{\ast}_3$, the set of distributions that satisfy mean independence, \ref{as:mean_iv}.  By the same sandwich argument used to prove Theorem \ref{theorem:same}, the identified set  for $\mathbb E_Q[Y_0]$ relative to $\mathbf{Q}$ corresponds to the identified set relative to $\mathbf{Q}^{\ast}_{1}$, and thus imposing this ARUM has no identifying power for $\mathbb E_Q[Y_0]$ beyond instrument exogeneity.

Finally, we show that there exists a $P$ for which $\mathbf Q_0(P, \mathbf Q) \neq \emptyset$ and \eqref{eq:ARUM-3arm-Y0} is strictly smaller than \eqref{eq:bounds1}, so that $\Theta_0(P, \mathbf Q)$ is not given by \eqref{eq:bounds1}. We do so by providing a numerical example. Consider the $P$ specified in Table \ref{tab:distP_Ex53} and the $Q_{\text{arum},\text{min}}$ and $Q_{\text{arum},\text{max}}$ specified in Tables \ref{tab:distQ_Ex53_arummin} and \ref{tab:distQ_Ex53_arummax} respectively, where we write $q(y_0 y_1 y_2, d_0 d_1) = Q\{Y_d = y_d, D_z = d_z, ~(d,z) \in \mathcal D \times \mathcal Z\}$ and omit any $q(\cdot)=0$. One can check that both $Q_{\text{arum},\text{min}}$ and $Q_{\text{arum},\text{max}}$ rationalize $P$ and satisfy Assumption \ref{as:exog}. One can further check that both $Q_{\text{arum},\text{min}}$ and $Q_{\text{arum},\text{max}}$ satisfy the restriction in \eqref{eq:arum-violation-rt}, so that $\mathbf Q_0(P, \mathbf Q) \neq \emptyset$. Evaluating \eqref{eq:ARUM-3arm-Y0} at $P$ gives the identified set for $\mathbb E_Q[Y_0]$ relative to $\mathbf Q$ as $[0.2518, 0.8167]$. One can further check that the two endpoints are attained by $\mathbb E_{Q_{\text{arum},\text{min}}}[Y_0] = 0.2518$ and $\mathbb E_{Q_{\text{arum},\text{max}}}[Y_0] = 0.8167$. On the other hand, if one evaluates \eqref{eq:bounds1} by setting $z^\ast(0) = 0,1$ at the same $P$, the resulting bounds for $\mathbb E_Q[Y_0]$ equal $[p_{10|0}, 1 - p_{00|0}] =  [0.2518, 0.8937]$ and $[p_{10|1}, 1 - p_{00|1}] = [0.2372, 0.8167]$ respectively. In both cases, \eqref{eq:ARUM-3arm-Y0} is strictly contained in \eqref{eq:bounds1}. Hence, $\Theta_0(P, \mathbf Q)$ is not given by \eqref{eq:bounds1}.

\begin{table}[ht!]
    \centering
    \small
    \begin{tabular}{|c|c|c|c|c|c|}
  \hline
  $p_{ 00|0 }$ &
  $p_{ 10|0 }$ &
  $p_{ 01|0 }$ &
  $p_{ 11|0 }$ &
  $p_{ 02|0 }$ &
  $p_{ 12|0 }$ \\
  0.1063 &
  0.2518 &
  0.2946 &
  0.3183 &
  0.0020 &
  0.0270 \\
  \hline
  $p_{ 00|1 }$ &
  $p_{ 10|1 }$ &
  $p_{ 01|1 }$ &
  $p_{ 11|1 }$ &
  $p_{ 02|1 }$ &
  $p_{ 12|1 }$ \\
  0.1833 &
  0.2372 &
  0.0140 &
  0.1399 &
  0.1701 &
  0.2555 \\ 
  \hline
    \end{tabular}
    \caption{Distribution $P$ in Appendix \ref{sec:arum-violation}.}
    \label{tab:distP_Ex53}
\end{table}

\begin{table}[ht!]
    \centering
    \small
    \begin{tabular}{|c|c|c|c|c|c|}
    \hline
         $q( 000,10 )$ & $q( 000,11 )$ & $q( 000,12 )$ & $q( 000,22 )$ & $q( 001,02 )$ & $q( 001,12 )$ \\
         0.0049 & 0.0140 & 0.1535 & 0.0020 & 0.1063 & 0.1222 \\
         \hline
         $q( 001,22 )$ & $q( 010,10 )$ & $q( 010,11 )$ & $q( 100,00 )$ & $q( 100,02 )$ & \\ 
         0.0270 & 0.1784 & 0.1399 & 0.2372 & 0.0146 & \\
    \hline
    \end{tabular}
    \caption{Distribution $Q_{\text{arum},\text{min}}$.}
    \label{tab:distQ_Ex53_arummin}
\end{table}

\begin{table}[ht!]
    \centering
    \small
    \begin{tabular}{|c|c|c|c|c|c|}
    \hline
         $q( 000,00 )$ & $q( 010,10 )$ & $q( 100,00 )$ & $q( 100,10 )$ & $q( 100,11 )$ & $q( 100,02 )$ \\       
         0.1063 & 0.0770 & 0.0837 & 0.0521 & 0.0140 & 0.1681 \\
         \hline
         $q( 100,22 )$ & $q( 101,12 )$ & $q( 101,22 )$ & $q( 110,10 )$ & $q( 110,11 )$ & \\
         0.0020 & 0.2285 & 0.0270 & 0.1014 & 0.1399 & \\
    \hline
    \end{tabular}
    \caption{Distribution $Q_{\text{arum},\text{max}}$.}
    \label{tab:distQ_Ex53_arummax}
\end{table}

 \section{Additional Examples of Models That Satisfy Assumption \ref{as:encouragement}} 
\label{sec:ap_ex}

In Section \ref{sec:ex}, we considered examples of restrictions on potential treatments previously considered in the literature that satisfy generalized monotonicity.  We now consider three additional such examples.

\begin{example}\label{ex:binaryarum}
Consider the ARUM of Example \ref{ex:arum} when $|\mathcal{D}|=2$, and let $\mathbf{Q}$ denote the set of distributions defined in that example. Then, Assumption  \ref{as:encouragement} holds for all $Q \in \mathbf{Q}$.  To see this, consider $Q \in \mathbf Q$.  Label $\mathcal{D}=\{0,1\}$, and let $g_{10}(z) = g(z,1)-g(z,0)$ and $U_{10} = U_1-U_0$. The assumptions of Example \ref{ex:arum} on $(U_1,U_0)$ imply that the distribution of $U_{10}$ is absolutely continuous with respect to Lebesgue measure and that $U_{10} \indep Z$. Ignoring ties that occur with probability zero, \eqref{eq:ARUM} can be rewritten as
\begin{equation} \label{eq:binarychoice}
D_z = \mathbbm{1}\{     g_{10}(z) + U_{10} \ge 0 \}~.
\end{equation}
Let $\overline{\mathcal Z} = \argmax_{ z \in \mathcal{Z} }\{g_{10}(z)\}$, and let  $\underline{\mathcal Z} = \argmin_{ z \in \mathcal{Z} }\{g_{10}(z)\}$. Then, $Q$ satisfies Assumption \ref{as:encouragement} with $\mathcal Z^\ast(1) = \overline{\mathcal Z}$ and  $\mathcal Z^\ast(0) = \underline{\mathcal Z}$.  To contrast with Example \ref{ex:arum}, note that \eqref{eq:uniform} holds if and only if $\overline{\mathcal Z}$ and $\underline{\mathcal{Z}} $ are both singletons. 
\end{example}


\begin{example}
\cite{kline2016evaluating} considers an RCT with a ``close substitute'' to study the effects of preschooling on educational outcomes. In their setting, $D \in \mathcal{D} = \{0,1,2\}$, where $D=0$ denotes home care (no preschool), $D = 2$ denotes a preschool program called Head Start, and $D = 1$ denotes  preschools other than Head Start, namely the close substitute. Let $Z \in \mathcal{Z}=\{0,1\}$ denote an indicator variable for an offer to attend Head Start. Assumption \ref{as:exog} holds because $Z$ is randomly assigned. \cite{kline2016evaluating} impose the restriction that
\begin{equation} \label{eq:kw}
Q \{ D_1 = 2 \mid D_0 \ne D_1\}=1~.
\end{equation}
The condition in \eqref{eq:kw} states that if the choice of a family changes upon receiving a Head Start offer, then they must choose Head Start when receiving the offer. In other words, it cannot be the case that upon receiving a Head Start offer, a family switches from no preschool to preschools other than Head Start, or the other way around. Assumption \ref{as:encouragement} then holds with $z^\ast(0)=z^\ast(1)=0$ and $z^\ast(2)=1$.  To see this, note that \eqref{eq:kw} implies $Q\{  D_0 \ne D_1, D_1 \ne 2\}=0$ and thus
\begin{align*}
    Q\{ D_0 \ne 0, D_1 =0\} & =   Q\{ D_0 \ne D_1 , D_1 =0\} =0~,   \\
    Q\{ D_0 \ne 1, D_1 =1\} & =   Q\{ D_0 \ne D_1 , D_1 =1\} =0~,   \\
    Q\{ D_1 \ne 2, D_0 =2\} & \le  Q\{ D_0 \ne D_1, D_1 \ne 2\} =0~.
\end{align*}
Note in this example Assumption \ref{as:encouragement} still holds although $|\mathcal Z| < |\mathcal D|$. See \cite{bai2025sharp} for results on the sharp testable implications of the assumptions for this example and Example \ref{ex:klm}.
\end{example}

\begin{example} \label{ex:klm}
\cite{kirkeboen2016field} study the effects of fields of study on earnings. In their setting, $\mathcal D = \{0, 1, 2\}$ represent three fields of study, ordered by their (soft) admission cutoffs from the lowest to the highest. The instrument is $Z \in \{0,1,2\}$, with $Z = 1$ when the student crosses the (soft) admission cutoff for field 1, $Z = 2$ when the student crosses the (soft) admission cutoff for field 2, and $Z = 0$ otherwise. The authors assume that $Z$ is exogenous in the sense that $Q$ satisfies Assumption \ref{as:exog} and impose the following monotonicity conditions:
\begin{align}
    \label{eq:monotonicity1} Q \{D_1 = 1 \mid D_0 = 1\} & = 1~, \\
    \label{eq:monotonicity2} Q \{D_2 = 2 \mid D_0 = 2\} & = 1~.
\end{align}
The conditions in \eqref{eq:monotonicity1}--\eqref{eq:monotonicity2} require that crossing the cutoff for field 1 or 2 weakly encourages them towards that field. They further impose the following ``irrelevance'' conditions:
\begin{align}
    \label{eq:irrelevance1} Q \{\mathbbm 1\{D_1 = 2\} & = \mathbbm 1\{D_0 = 2\} \mid D_0 \neq 1, D_1 \neq 1\} = 1~, \\
    \label{eq:irrelevance2} Q \{\mathbbm 1\{D_2 = 1\} & = \mathbbm 1\{D_0 = 1\} \mid D_0 \neq 2, D_2 \neq 2\} = 1~.
\end{align}
The condition in \eqref{eq:irrelevance1} states that if crossing the cutoff for field 1 does not cause the student to switch to field 1, then it does not cause them to switch to or away from field 2. A similar interpretation applies to \eqref{eq:irrelevance2}. \cite{lee2023treatment} show the set of all distributions that satisfy \eqref{eq:monotonicity1}--\eqref{eq:irrelevance2} are equivalent to a strict one-to-one targeting model with $|\mathcal Z| = 3$ and $|\mathcal D^\dagger| = 2$; it therefore follows from Remark \ref{ex:lee2023} that any $Q$ that satisfies \eqref{eq:monotonicity1}--\eqref{eq:irrelevance2} also satisfies Assumption \ref{as:encouragement}. Here, we establish directly that \eqref{eq:monotonicity1}--\eqref{eq:irrelevance2} imply Assumption \ref{as:encouragement} with $z^\ast(0) = 0$, $z^\ast(1) = 1$, and $z^\ast(2) = 2$. To show $z^\ast(0) = 0$, we prove by contradiction that
\[ Q \{D_0 \neq 0, D_1 = 0\} = 0~. \]
Suppose with positive probability that $D_0 \neq 0$ but $D_1 = 0$. On this event, \eqref{eq:monotonicity1} implies $D_0 \neq 1$, so $D_0 = 2$. But $D_1 = 0$, which contradicts \eqref{eq:irrelevance1}. Similarly,
\[ Q \{D_0 \neq 0, D_2 = 0\} = 0~, \]
and therefore $z^\ast(0) = 0$. To show $z^\ast(1) = 1$, first note \eqref{eq:monotonicity1} implies
\[ Q \{D_1 \neq 1, D_0 = 1\} = 0~. \]
It therefore remains to argue by contradiction that
\begin{equation} \label{eq:klm-contra-1}
    Q \{D_1 \neq 1, D_2 = 1\} = 0~.
\end{equation}
Suppose with positive probability that $D_1 \neq 1$ but $D_2 = 1$. On this event, \eqref{eq:monotonicity1} implies $D_0 \neq 1$. If $D_0 = 2$,  then \eqref{eq:monotonicity2} implies $D_2 = 2$, a contradiction to $D_2 = 1$; if instead $D_0 = 0$, then because we assume $D_2 = 1$, \eqref{eq:irrelevance2} implies $D_2 \neq 1$, another contradiction. Therefore, \eqref{eq:klm-contra-1} holds, and $z^\ast(1) = 1$. $z^\ast(2) = 2$ can be established following similar arguments.
\end{example}

\newpage
\bibliography{multivalued}

\end{document}